\documentclass[11pt,a4paper]{article}
\usepackage[left=2cm,right=2cm]{geometry}
\usepackage[T1]{fontenc}
\usepackage{mdwlist}
\usepackage[usenames,dvipsnames]{xcolor}
\usepackage[all,dvips,color,arrow,curve,poly]{xy}
\UseCrayolaColors

\usepackage{amsfonts}
\usepackage{url}

\newtheorem{theorem}{Theorem}[section]
\newtheorem{lemma}[theorem]{Lemma}
\newtheorem{corollary}[theorem]{Corollary}
\newtheorem{observation}[theorem]{Observation}
\newtheorem{proposition}[theorem]{Proposition}
\newtheorem{definition}[theorem]{Definition}
\newenvironment{proof}{\paragraph{Proof:}}{\hfill$\square$}
\newcommand{\calS}{\mathcal{S}}
\newcommand{\calA}{\mathcal{A}}
\newcommand{\calI}{\mathcal{I}}
\newcommand{\calJ}{\mathcal{J}}
\newcommand{\calL}{\mathcal{L}}
\newcommand{\calO}{\mathcal{O}}
\newcommand{\reals}{\ensuremath{\mathbb{R}}}
\newcommand{\NP}{NP}

\usepackage{url}

\usepackage{authblk}

\title{Max Point-Tolerance Graphs\thanks{Part of this work was previously presented at the 4th biennial Canadian Discrete and Algorithmic Mathematics Conference (CanaDAM) in St. John's, NL, Canada June 10-13, 2013. The abstract and slides are available at: \protect\url{http://canadam.math.ca/2013/program/abs/grg2.html\#sc}.}
}

\author[1]{Daniele Catanzaro}
\author[2]{Steven Chaplick}
\author[2]{Stefan Felsner}	
\author[4]{Bjarni V. Halld\'{o}rsson}
\author[5]{Magn\'{u}s M. Halld\'{o}rsson}
\author[2]{Thomas Hixon}
\author[6]{Juraj Stacho}

\affil[1]{Louvain School of Management and Center for Operations Research and Econometrics (CORE), Universit\'e Catholique de Louvain, Mons, Belgium}

\affil[2]{Institut f\"ur Mathematik, Technische Universit\"at Berlin, Berlin, Germany}

\affil[4]{School of Science and Engineering, Reykjavik University, Reykjav\'ik, Iceland}

\affil[5]{ICE-TCS, School of Computer Science, Reykjavik University, Reykjav\'{i}k, Iceland}

\affil[6]{Department of Industrial Engineering and Operations Research,
Columbia University, New York NY, United States}

\begin{document}
\maketitle

\begin{abstract}
A graph $G$ is a \emph{max point-tolerance (MPT)} graph if each vertex $v$ of $G$ can
be mapped to a \emph{pointed-interval} $(I_v, p_v)$ where $I_v$ is an interval
of $\reals$ and $p_v \in I_v$ such that $uv$ is an edge of $G$ iff $I_u \cap I_v
\supseteq \{p_u, p_v\}$.
MPT graphs model relationships among DNA fragments in genome-wide association studies as well as basic transmission problems in telecommunications.  
We formally introduce this graph class, characterize it, study combinatorial optimization problems on it, and relate it to several well known graph classes. We characterize MPT graphs as a special case of several 2D geometric intersection graphs; namely, triangle, rectangle, L-shape, and line segment intersection graphs. We further characterize MPT as having certain linear orders on their vertex set. Our last characterization is that MPT graphs are precisely obtained by intersecting special pairs of interval graphs. 
We also show that, on MPT graphs, the maximum weight independent set problem can be solved in polynomial time, the coloring problem is \NP-complete, and the clique cover problem has a 2-approximation. 
Finally, we demonstrate several connections to known graph classes; e.g., MPT graphs strictly contain interval graphs and outerplanar graphs, but are incomparable to permutation, chordal, and planar graphs.

\end{abstract}



\section{Introduction}

\emph{Interval} graphs (namely, the intersection graphs of intervals on a line) are well-studied in computer science and discrete mathematics (see e.g.,\cite{Fishburn1985-interval,FG1965}). 
Many combinatorial problems which are \NP-hard in general can be solved efficiently when restricted to interval graphs. For example, the maximum clique problem \cite{FG1965}, the maximum weight independent set problem \cite{Frank1975-chordal-wis}, and the coloring problem \cite{G-AlgGraphTheory2004} can all be solved in linear time on interval graphs.
The recognition problem is also solvable in linear time \cite{BL1976-PQ}.

Due to their theoretical and practical significance many generalizations of interval graphs have been studied (see e.g.,\cite{ACGLLS2012-VPG,CK1987-gen_interval,Gavril1978-VPT,GM1982-tolerance}). 
Particularly relevant to this work are \emph{tolerance graphs}, first introduced in \cite{GM1982-tolerance}. 
A graph is a \emph{tolerance} graph (also known as a \emph{min tolerance} graph) when every vertex $v$ of $G$ can be associated with an interval $I_v$ (of the real number line: $\reals$) and a tolerance value $t_v \in \mathbb{R}$ such that $uv$ is an edge of $G$ iff $|I_u \cap I_v| \geq \min\{t_u,t_v\}$. 
Similarly, a graph is a \emph{max tolerance} graph when each vertex $v$ of $G$ can be associated with an interval $I_v$ and tolerance  $t_v$ such that $uv \in E(G)$ iff $|I_u \cap I_v|\geq \max\{t_u,t_v\}$. For a detailed study of tolerance graphs see  \cite{GT-Tolerance2004}.

In this paper we introduce the class of \emph{max point-tolerance (MPT)} graphs\footnote{Using the phrasing of Golumbic and Trenk~\cite{GT-Tolerance2004} this class would be called max point-core bitolerance graphs. However, this particular class of tolerance graphs was not discussed in \cite{GT-Tolerance2004}.}.
A graph $G$ is an MPT graph when each vertex $v$ of $G$ can be represented by an interval $I_v$ of $\reals$ together with a point $p_v \in I_v$ such that two vertices $u,v$ are adjacent iff both $p_u$ and $p_v$ belong to $I_u \cap I_v$; i.e., each pair of intervals can ``tolerate'' a non-empty intersection (without forming an edge) as long as at least one distinguished point is not contained in this intersection. 
We call such a collection $\{(I_v,p_v)\}_{v \in V(G)}$ of pointed intervals an MPT representation of $G$. 
Moreover, we also denote each $(I_v,p_v)$ by a triplet $(s_v,p_v,e_v)$ where $s_v$ and $e_v$ denote the start and end of $I_v$ respectively. 

MPT graphs have a number of practical applications.  
They can be used to detect loss of heterozygosity events in the human genome; see e.g., \cite{HATI2011-point-tol3,Nature2008-point-tol}. 
In such applications an interval $I$ represents the maximal boundary on a chromosome region from an individual that may carry a deletion and the point $p$ represents a site in the considered region that shows evidence for a deletion.
MPT graphs could also be used to model telecommunication networks; e.g., communication where devices receive message on a wide channel (interval) and send messages on a narrow on a sub-band (point) of that channel. Such an asymmetric ``big'' downlink / ``small'' uplink model is quite common in telecommunication networks (see, e.g., \cite{Uplink-Downlink-2012,Uplink-Downlink-2001}). 
In this situation the edges of the MPT graph correspond to devices with direct two-way communication. 

Some classical optimization problems on MPT graphs correspond to practical problems.  
For example, when modeling genome-wide association studies, finding the chromosomal region showing the highest evidence for a massive loss of heterozygosity in a population of individuals involves solving the maximum clique problem and partitioning all evidence-of-deletion sites into the minimal number of deletions involves solving the minimum clique cover problem \cite{CLH2013-point-tol4}. 
In our telecommunications example, a minimum clique cover corresponds to partitioning the devices into a minimum collection of sets of fully-communicable devices.

Interestingly, the maximum weight clique problem on a MPT graph was shown to be polynomially solvable due to the fact
that an MPT graph can have at most $\calO(n^2)$ maximal cliques \cite{CLH2013-point-tol4}. 
Additionally, the minimum weight clique cover problem was shown to be \NP-complete for submodular cost functions \cite{CLH2013-point-tol4,DJQ2007-intervalMCP}. The complexity of the unweighted clique cover problem on MPT graphs remains unresolved. 

Finally, closely related to MPT graphs is the class of \emph{interval catch digraphs}. 
A digraph $D$ is an \emph{interval catch digraph} when each vertex $v$ of $D$ can be mapped to an interval $I_v$ of $\reals$ together with a point $p_v \in I_v$ such that $uv$ is an arc of $D$ iff $p_u \in I_v$. Notice that MPT graphs are precisely the underlying undirected graphs of the symmetric edges of interval catch digraphs. 
Interval catch digraphs have a vertex order characterization \cite{Maehara1984-CatchDigraphs}, an asteroidal-triple characterization \cite{Prisner1989-CatchDigraphs}, and a polynomial time recognition algorithm \cite{Prisner1994-CatchDigraphs}. However, these results do not translate to MPT graphs.

\medskip

\textbf{Our Contributions:}
We provide characterizations of MPT graphs, utilize these characterizations for combinatorial optimization problems, and relate MPT graphs to well-known graph classes.

In section \ref{sec:Ls} we characterize MPT graphs as a special case of \emph{L-graphs} (intersection graphs of L-shapes in the plane). 
This will imply that MPT is also a subclass of rectangle intersection graphs (also known as boxicity-2 graphs \cite{boxicity}) and of triangle intersection graphs. 
We also use this characterization to show that interval graphs and 2D ray graphs are strict subclasses of MPT graphs. 
We further characterize MPT graphs by certain linear vertex orders. 
In particular, we show that a graph $G = (V,E)$ is MPT iff the vertices of $G$ can be linearly ordered by $<$ so that no quadruple $u,v,w,x \in V$ with $u < v < w < x$ has the edges $uw$ and $vx$ without the edge $vw$. 
Related to this ordering condition, we also describe MPT graphs as the intersection of two special interval graphs (see
Theorem \ref{thm:2interval-MPT}). 
Finally, MPT graphs are characterized as intersection graphs of certain line segments from cyclic line arrangements. 

These characterizations are then used to study combinatorial optimization problems on MPT graphs.
Namely, we demonstrate that the \emph{weighted independent set (WIS)} problem can be solved in polynomial time, the clique cover problem can be 2-approximated in polynomial time, and that the \emph{coloring} problem is \NP-complete but can be $\log(n)$-approximated in polynomial time. 
As part of the approximations, we show that the clique cover number $\gamma(G)$ is at most twice the independence number $\alpha(G)$ and that the chromatic number $\chi(G)$ is at most $\calO(\omega \log(\omega))$ where $\omega$ is the clique number\footnote{The bound on $\chi(G)$ follows from \cite{Chalermsook2011} and one of our characterizations.}. 

Finally, we observe some structural results and compare MPT graphs to several well-known graph classes. For example, we observe that outerplanar graphs are a proper subclass of MPT graphs and characterize them by a ``contact'' MPT representation. We additionally observe infinite families of forbidden induced subgraphs for MPT graphs which are constructed from non-interval and non-outerplanar graphs. 

\medskip

\textbf{Related Work:}
While our results have been obtained independently, there are several places which 
overlap with some existing papers \cite{Lubiw1991,c-And2013,HittingRectangles2014}. 
We will identify each of these as they are presented. 
Note that \cite{c-And2013} is technical report, \cite{HittingRectangles2014} is a refereed conference paper, and \cite{Lubiw1991} is a journal publication.
Some of our results also appear in the Masters Thesis of 
our co-author Thomas Hixon \cite{Hixon2013}. 

\medskip

\textbf{Preliminaries:}
All graphs considered in this paper are simple, undirected, and loopless (unless
otherwise stated). For a graph $G$ with vertex set $V$ and edge set $E$, we use
the following notation. The symbols $n$ and $m$ denote $|V|$ and $|E|$
respectively. For a subset $S$ of $V$, $G[S]$ denotes the subgraph of $G$
induced by $S$ and $G \setminus S$ denotes the subgraph of $G$ induced by $V
\setminus S$; i.e., $G \setminus S = G[V \setminus S]$. For a vertex $v \in V$,
$N(v)$ denotes the neighborhood of $v$ (i.e., the vertices in $G$ which are
adjacent to $v$).

\section{Geometric representations of MPT graphs}
\label{sec:Ls}

In this section we relate MPT graphs to geometric intersection graphs.
Specifically, we characterize MPT graphs as intersection graphs of axis-aligned
L-shapes whose corner points form a line with negative slope (namely,
\emph{linear L-graphs} as defined below). Once we formalize this it will be easy
to see that this implies that MPT graphs are a special subclass of boxicity-2 
graphs and triangle intersection graphs. 
The equivalence between linear L-graphs and MPT graphs is also stated in \cite{c-And2013}.
Later in this paper we use these characterizations to study combinatorial optimization problems on MPT graphs and to relate MPT graphs to classical graph classes. 

An \emph{L-shape} consists of a vertical line segment and a horizontal line segment with a {\em corner} that is the lowest point of the vertical segment and the left-most point of the horizontal segment.
We define a \emph{linear L-system} $\mathcal{L}$ to be a collection of L-shapes \{$L_1$, \ldots, $L_n$\} in the plane such that the corner points of $L_1$, \ldots, $L_n$ are distinct and form a line with negative slope.  
We say that a graph $G$ is a \emph{linear L-graph} if $G$ is the intersection graph of a linear L-system $\calL$ and we refer to $\calL$ as a linear L-system \emph{of} $G$. 
We define \emph{linear rectangle graphs} and \emph{linear right-triangle graphs} similarly (i.e., with the lower-left corners of the shapes forming a line with negative slope; note: we always consider the lower-left corner of each triangle to be the right angle). 
In particular, it is easy to see that these three graph classes are the same; e.g., as in Figure \ref{fig:L-rect-tri}.  

Without loss of generality we assume that the corner points in all linear
systems have the form $(c,-c)$ for some positive integer $c$.  This allows us
to specify each L-shape $L$ in a linear L-system by $(t_L,c_L,r_L)$ where:
$-t_L$ is the y-coordinate of the top of $L$, $(c_L,-c_L)$ is the corner
point of $L$, and $r_L$ is the x-coordinate of the right-most point of $L$.
Such an L-shape is given in Figure \ref{fig:L-anatomy}. 

\begin{figure}[ht]
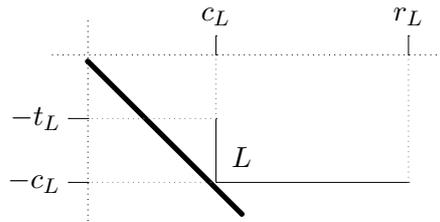
\centering
$\xy/r4pc/:{\ar@{-}@[|2pt] (0,-0.05);(1.2,-1.25)};{\ar@{-} (1,-0.5);(1,-1)};
{\ar@{-} (1,-1);(2.5,-1)};{\ar@{.}@[gray] (1,0);(1,-0.5)};
{\ar@{.}@[gray] (0,-0.5);(1,-0.5)};{\ar@{.}@[gray] (0,-1);(1,-1)};
{\ar@{.}@[gray] (2.5,0);(2.5,-1)};{\ar@{.} (-0.3,0);(2.8,0)};(1.2,-0.8)*{L};
{\ar@{.} (0,0.3);(0,-1.3)};{\ar@{-} (1,0);(1,0.15)};{\ar@{-} (2.5,0);(2.5,0.15)};
{\ar@{-} (0,-0.5);(-0.15,-0.5)};{\ar@{-} (0,-1);(-0.15,-1)};(-0.3,-0.5)*[l]{-t_L};
(-0.3,-1)*[l]{-c_L};(1,0.3)*{c_L};(2.5,0.3)*{r_L};\endxy$
\caption{Anatomy of an L-shape in a linear L-system. Notice that we include a
``platform'' corresponding to the line $x+y=0$ to emphasize the linearity of the
system.}
\label{fig:L-anatomy}
\end{figure}

\begin{figure}[ht]
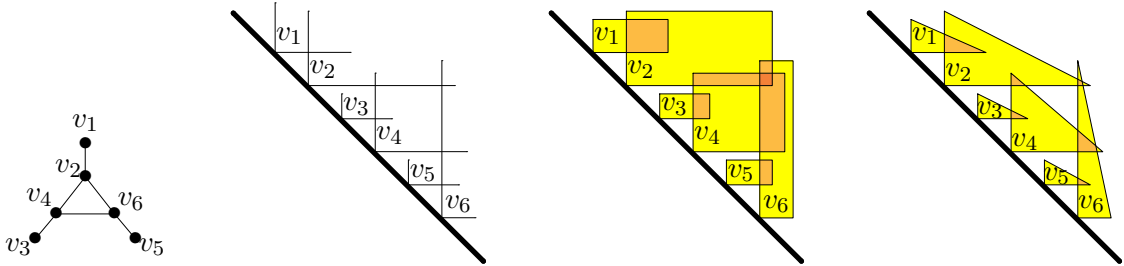
\centering
$\xy/r1.3pc/:(1.7,-3.2)*[o]{}="v1";
(1.7,-4)*[o]{}="v2";(0.5,-5.5)*[o]{}="v3";(1,-4.9)*[o]{}="v4";
(2.9,-5.5)*[o]{}="v5";(2.4,-4.9)*[o]{}="v6";
"v1"*{\bullet};"v2"*{\bullet};
"v3"*{\bullet};"v4"*{\bullet};"v5"*{\bullet};"v6"*{\bullet};{\ar@{-} "v1";"v2"};
{\ar@{-} "v2";"v4"};{\ar@{-} "v2";"v6"};{\ar@{-} "v4";"v6"};{\ar@{-} "v3";"v4"};
{\ar@{-} "v5";"v6"};"v1"+(0,0.5)*{v_1};"v2"+(-0.4,0.2)*{v_2};
"v3"+(-0.4,-0.2)*{v_3};"v4"+(-0.4,0.4)*{v_4};"v5"+(0.4,-0.2)*{v_5};
"v6"+(0.4,0.3)*{v_6};\endxy$
\qquad
$\xy/r1.3pc/:{\ar@{-}@[|2pt] (0,-0.05);(6,-6.05)};{\ar@{-} (1,0.2);(1,-1)};
{\ar@{-} (1,-1);(2.8,-1)};{\ar@{-} (1.8,0);(1.8,-1.8)};
{\ar@{-} (1.8,-1.8);(5.3,-1.8)};{\ar@{-} (2.6,-2);(2.6,-2.6)};
{\ar@{-} (2.6,-2.6);(3.8,-2.6)};{\ar@{-} (3.4,-1.5);(3.4,-3.4)};
{\ar@{-} (3.4,-3.4);(5.6,-3.4)};{\ar@{-} (4.2,-3.6);(4.2,-4.2)};
{\ar@{-} (4.2,-4.2);(5.4,-4.2)};{\ar@{-} (5,-1.2);(5,-5)};
{\ar@{-} (5,-5);(5.8,-5)};(1.35,-0.6)*{v_1};(2.15,-1.4)*{v_2};(2.95,-2.3)*{v_3};
(3.75,-3.0)*{v_4};(4.55,-3.9)*{v_5};(5.4,-4.65)*{v_6};\endxy$
\qquad
$\xy/r1.3pc/:
@i@={(1,-0.2),(1,-1),(2.8,-1),(2.8,-0.2)},0*[yellow]\xypolyline{*};
@i@={(1.8,0),(1.8,-1.8),(5.3,-1.8),(5.3,0)},0*[yellow]\xypolyline{*};
@i@={(2.6,-2),(2.6,-2.6),(3.8,-2.6),(3.8,-2),(2.6,-2)},0*[yellow]\xypolyline{*};
@i@={(3.4,-1.5),(3.4,-3.4),(5.6,-3.4),(5.6,-1.5),(3.4,-1.5)},0*[yellow]\xypolyline{*};
@i@={(4.2,-3.6),(4.2,-4.2),(5.3,-4.2),(5.3,-3.6),(4.2,-3.6)},0*[yellow]\xypolyline{*};
@i@={(5,-1.2),(5,-5),(5.8,-5),(5.8,-1.2),(5,-1.2)},0*[yellow]\xypolyline{*};
@i@={(1.8,-0.2),(1.8,-1),(2.8,-1),(2.8,-0.2)},0*[Dandelion]\xypolyline{*};
@i@={(3.4,-1.5),(3.4,-1.8),(5,-1.8),(5,-1.5)},0*[Dandelion]\xypolyline{*};
@i@={(5,-1.5),(5,-1.8),(5.3,-1.8),(5.3,-1.5)},0*[Orange]\xypolyline{*};
@i@={(5.3,-1.5),(5.3,-1.8),(5,-1.8),(5,-3.4),(5.6,-3.4),(5.6,-1.5),
(5.3,-1.5)},0*[Dandelion]\xypolyline{*};
@i@={(5,-1.2),(5,-1.5),(5.3,-1.5),(5.3,-1.2)},0*[Dandelion]\xypolyline{*};
@i@={(5,-3.6),(5,-4.2),(5.3,-4.2),(5.3,-3.6)},0*[Dandelion]\xypolyline{*};
@i@={(3.4,-2),(3.4,-2.6),(3.8,-2.6),(3.8,-2)},0*[Dandelion]\xypolyline{*};
@i@={(1,-0.2),(1,-1),(2.8,-1),(2.8,-0.2),(1,-0.2)},0*\xypolyline{};
@i@={(1.8,0),(1.8,-1.8),(5.3,-1.8),(5.3,0),(1.8,0)},0*\xypolyline{};
@i@={(2.6,-2),(2.6,-2.6),(3.8,-2.6),(3.8,-2),(2.6,-2)},0*\xypolyline{};
@i@={(3.4,-1.5),(3.4,-3.4),(5.6,-3.4),(5.6,-1.5),(3.4,-1.5)},0*\xypolyline{};
@i@={(4.2,-3.6),(4.2,-4.2),(5.3,-4.2),(5.3,-3.6),(4.2,-3.6)},0*\xypolyline{};
@i@={(5,-1.2),(5,-5),(5.8,-5),(5.8,-1.2),(5,-1.2)},0*\xypolyline{};
{\ar@{-}@[|2pt] (0,-0.05);(6,-6.05)};(1.35,-0.6)*{v_1};(2.15,-1.4)*{v_2};
(2.95,-2.3)*{v_3};(3.75,-3.0)*{v_4};(4.55,-3.9)*{v_5};(5.4,-4.65)*{v_6};\endxy$
\qquad
$\xy/r1.3pc/:@i@={(1,-0.2),(1,-1),(2.8,-1)},0*[yellow]\xypolyline{*};
@i@={(1.8,0),(1.8,-1.8),(5.3,-1.8)},0*[yellow]\xypolyline{*};
@i@={(2.6,-2),(2.6,-2.6),(3.8,-2.6)},0*[yellow]\xypolyline{*};
@i@={(3.4,-1.5),(3.4,-3.4),(5.6,-3.4)},0*[yellow]\xypolyline{*};
@i@={(4.2,-3.6),(4.2,-4.2),(5.3,-4.2)},0*[yellow]\xypolyline{*};
@i@={(5,-1.2),(5,-5),(5.8,-5)},0*[yellow]\xypolyline{*};
@i@={(1.8,-0.55556),(1.8,-1),(2.8,-1)},0*[Dandelion]\xypolyline{*};
@i@={(3.4,-1.5),(3.4,-1.8),(3.74737,-1.8)},0*[Dandelion]\xypolyline{*};
@i@={(3.4,-2.4),(3.4,-2.6),(3.8,-2.6)},0*[Dandelion]\xypolyline{*};
@i@={(5,-4.0364),(5,-4.2),(5.3,-4.2)},0*[Dandelion]\xypolyline{*};
@i@={(5,-2.8818),(5,-3.4),(5.46316,-3.4),(5.43268,-3.25556)},0*[Dandelion]\xypolyline{*};
@i@={(5,-1.6457),(5,-1.8),(5.12632,-1.8),(5.10523,-1.6998)},0*[Dandelion]\xypolyline{*};
@i@={(1,-0.2),(1,-1),(2.8,-1),(1,-0.2)},0*\xypolyline{};
@i@={(1.8,0),(1.8,-1.8),(5.3,-1.8),(1.8,0)},0*\xypolyline{};
@i@={(2.6,-2),(2.6,-2.6),(3.8,-2.6),(2.6,-2)},0*\xypolyline{};
@i@={(3.4,-1.5),(3.4,-3.4),(5.6,-3.4),(3.4,-1.5)},0*\xypolyline{};
@i@={(4.2,-3.6),(4.2,-4.2),(5.3,-4.2),(4.2,-3.6)},0*\xypolyline{};
@i@={(5,-1.2),(5,-5),(5.8,-5),(5,-1.2)},0*\xypolyline{};
{\ar@{-}@[|2pt] (0,-0.05);(6,-6.05)};(1.35,-0.6)*{v_1};(2.15,-1.48)*{v_2};
(2.9,-2.38)*{v_3};(3.73,-3.15)*{v_4};(4.5,-3.98)*{v_5};(5.33,-4.65)*{v_6};\endxy$
\caption{(from left-to-right) The net $G$, a linear L-system $\calL$ of $G$, the
linear rectangle-system corresponding to $\calL$, and the linear
right-triangle-system corresponding to $\calL$.}
\label{fig:L-rect-tri}
\end{figure}

\begin{theorem}\label{thm:MPT=L}
Max point-tolerance graphs are precisely linear L-graphs.
\end{theorem}

\begin{proof}
Let $\{(s_1,p_1,e_1), \ldots, (s_n,p_n,e_n)\}$ be a MPT representation of a graph
$G$. Consider the linear L-system $\calL = \{L_1, \ldots, L_n\}$ where $t_{L_i} =
-s_i$, $c_{L_i} = p_i$, and $r_{L_i} = e_i$. 
The theorem follows from the depiction of this construction given in Figure \ref{fig:MPT-L}
\end{proof}

\begin{figure}[ht]
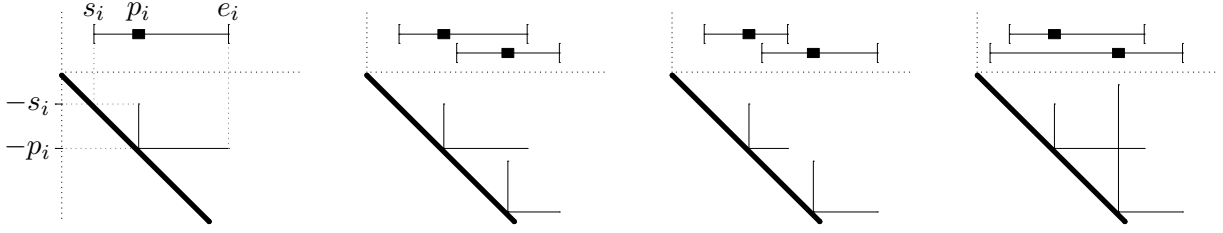
\centering
$\xy/r2pc/:{\ar@{-} (1,-0.3);(1,-1)};{\ar@{-} (1,-1);(2.4,-1)};
{\ar@{-} (0.3,0.8);(2.4,0.8)};{\ar@{-} (0.3,0.95);(0.3,0.65)};
{\ar@{-} (2.4,0.95);(2.4,0.65)};{\ar@{.} (-0.2,0.2);(3.5,0.2)};
{\ar@{.} (-0.2,1.2);(-0.2,-2.1)};(1,0.8)*[F**:black]{\phantom{_s}};
{\ar@{.}@[grey] (0.3,0.65);(0.3,-0.3)};{\ar@{.}@[grey] (-0.2,-0.3);(1,-0.3)};
{\ar@{-} (-0.2,-0.3);(-0.3,-0.3)};{\ar@{.}@[grey] (2.4,0.65);(2.4,-1)};
{\ar@{-} (-0.2,-1);(-0.3,-1)};{\ar@{.}@[grey] (-0.2,-1);(0.95,-1)};
(-0.55,-0.3)*[l]{-s_i};(-0.55,-1)*[l]{-p_i};(0.3,1.15)*{s_i};(1,1.15)*{p_i};
(2.4,1.15)*{e_i};{\ar@{-}@[|2pt] (-0.2,0.15);(2.1,-2.15)};\endxy$
\qquad
$\xy/r2pc/:{\ar@{-}@[|2pt] (-0.2,0.15);(2.1,-2.15)};{\ar@{-} (1,-0.3);(1,-1)};
{\ar@{-} (1,-1);(2.3,-1)};{\ar@{-} (2,-1.2);(2,-2)};{\ar@{-} (2,-2);(2.8,-2)};
{\ar@{-} (0.3,0.8);(2.3,0.8)};{\ar@{-} (0.3,0.95);(0.3,0.65)};
{\ar@{-} (2.3,0.95);(2.3,0.65)};{\ar@{-} (1.2,0.5);(2.8,0.5)};
{\ar@{-} (1.2,0.65);(1.2,0.35)};{\ar@{-} (2.8,0.65);(2.8,0.35)};
{\ar@{.} (-0.2,0.2);(3.5,0.2)};{\ar@{.} (-0.2,1.2);(-0.2,0.2)};
(1,0.8)*[F**:black]{\phantom{_s}};(2,0.5)*[F**:black]{\phantom{_s}};\endxy$
\qquad
$\xy/r2pc/:{\ar@{-}@[|2pt] (-0.2,0.15);(2.1,-2.15)};{\ar@{-} (1,-0.3);(1,-1)};
{\ar@{-} (1,-1);(1.6,-1)};{\ar@{-} (2,-1.2);(2,-2)};{\ar@{-} (2,-2);(3.0,-2)};
{\ar@{-} (0.3,0.8);(1.6,0.8)};{\ar@{-} (0.3,0.95);(0.3,0.65)};
{\ar@{-} (1.6,0.95);(1.6,0.65)};{\ar@{-} (1.2,0.5);(3.0,0.5)};
{\ar@{-} (1.2,0.65);(1.2,0.35)};{\ar@{-} (3,0.65);(3,0.35)};
{\ar@{.} (-0.2,0.2);(3.5,0.2)};{\ar@{.} (-0.2,1.2);(-0.2,0.2)};
(1,0.8)*[F**:black]{\phantom{_s}};(2,0.5)*[F**:black]{\phantom{_s}};\endxy$
\qquad
$\xy/r2pc/:{\ar@{-}@[|2pt] (-0.2,0.15);(2.1,-2.15)};{\ar@{-} (1,-0.3);(1,-1)};
{\ar@{-} (1,-1);(2.4,-1)};{\ar@{-} (2,0);(2,-2)};{\ar@{-} (2,-2);(3.0,-2)};
{\ar@{-} (0.3,0.8);(2.4,0.8)};{\ar@{-} (0.3,0.95);(0.3,0.65)};
{\ar@{-} (2.4,0.95);(2.4,0.65)};{\ar@{-} (0,0.5);(3.0,0.5)};
{\ar@{-} (0,0.65);(0,0.35)};{\ar@{-} (3,0.65);(3,0.35)};
{\ar@{.} (-0.2,0.2);(3.5,0.2)};{\ar@{.} (-0.2,1.2);(-0.2,0.2)};
(1,0.8)*[F**:black]{\phantom{_s}};(2,0.5)*[F**:black]{\phantom{_s}};\endxy$
\caption{Illustrating the equivalence between MPT representations and linear
L-systems. From left-to-right: the L-shape corresponding to a pointed-interval,
two examples of non-adjacent vertices as pointed-intervals and the corresponding
linear Ls, and one example of adjacent vertices as pointed-intervals and the
corresponding linear Ls.}
\label{fig:MPT-L}
\end{figure}

\section{L-systems of Interval Graphs} 
\label{sec:interval}

In this section we connect interval graphs with MPT graphs. 
We do this by demonstrating that every interval representation of a
graph is equivalent to an \emph{anchored} linear L-system (see Definition
\ref{def:anchored} and Proposition~\ref{pro:intervals_as_Ls}). 
Interval graphs have are also observed to be a subclass of MPT graphs
in \cite{c-And2013}. In fact, they claim that rooted path graphs (a superclass 
of interval graphs) are a subclass of MPT graphs, but they do not 
observe our characterization. 
We later use this characterization in our 2-approximation of clique cover
and to identify an infinite family of non-MPT graphs. 

\begin{definition}\label{def:anchored}
A linear L-system $\calL$ is \emph{anchored} if there exists $A \in \reals$ such
that $t_{L} \leq A \leq c_{L}$ for every $L \in \calL$. Note: we say that
$\calL$ is \emph{anchored} at $A$ and refer to $A$ as the \emph{anchor point} of
$\calL$. 
\end{definition}

\begin{proposition}\label{pro:intervals_as_Ls}
$G = (V,E)$ is an interval graph iff $G$ has an anchored linear L-system.
\end{proposition}
\begin{proof}
$(\Longrightarrow)$ Let $\calI = \{I_1, \ldots, I_n\}$ be an interval
representation of $G$ where $s_{i}$ and $e_{i}$ denote the starting and
ending points of the interval $I_i$ (respectively) for each $i \in \{1, \ldots,
n\}$. Furthermore, (wlog) assume $s_{i}\geq 0$ and $s_{i} < s_{j}$ iff $i < j$. Consider
the linear L-system $\calL = \{L_1, \ldots, L_n\}$ such that $L_i =
(0,s_{i},e_{i})$; i.e., $\calL$ is anchored at 0. Notice that, when two
intervals $I_i,I_j$ (1 $\leq i < j \leq n$) intersect, the corresponding
L-shapes $L_i,L_j$ will also intersect. Specifically, the horizontal segment of
$L_i$ will intersect the vertical segment of $L_j$ (see Figure
\ref{fig:interval-L} (left)). Moreover, when two intervals are disjoint the
corresponding L-shapes will be disjoint since their horizontal segments will not
have any common x-coordinates (see Figure \ref{fig:interval-L} (right)). 

\begin{figure}[ht]
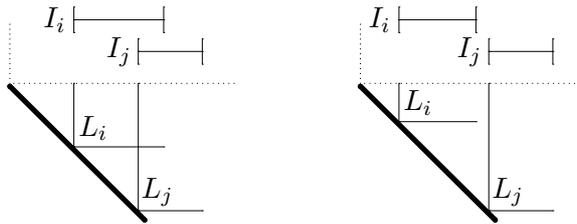
\centering
$\xy/r2pc/:{\ar@{-}@[|2pt] (0,-0.05);(2.1,-2.15)};
{\ar@{-} (1,0);(1,-1)};{\ar@{-} (1,-1);(2.4,-1)};{\ar@{-} (2,0);(2,-2)};
{\ar@{-} (2,-2);(3.0,-2)};{\ar@{.} (0,0);(3.5,0)};{\ar@{.} (0,1);(0,0)};
{\ar@{-} (1,1);(2.4,1)};{\ar@{-} (1,1.2);(1,0.8)};{\ar@{-} (2.4,1.2);(2.4,0.8)};
{\ar@{-} (2,0.5);(3.0,0.5)};{\ar@{-} (2,0.7);(2,0.3)};{\ar@{-} (3,0.7);(3,0.3)};
(0.7,1)*[l]{I_i};(1.7,0.5)*[l]{I_j};(1.3,-0.7)*{L_i};(2.3,-1.7)*{L_j};\endxy$
\qquad\qquad
$\xy/r2pc/:{\ar@{-}@[|2pt] (0,-0.05);(2.1,-2.15)};
{\ar@{-} (0.6,0);(0.6,-0.6)};{\ar@{-} (0.6,-0.6);(1.8,-0.6)};
{\ar@{-} (2,0);(2,-2)};{\ar@{-} (2,-2);(3.0,-2)};{\ar@{.} (0,0);(3.5,0)};
{\ar@{.} (0,1);(0,0)};{\ar@{-} (0.6,1);(1.8,1)};{\ar@{-} (0.6,1.2);(0.6,0.8)};
{\ar@{-} (1.8,1.2);(1.8,0.8)};{\ar@{-} (2,0.5);(3.0,0.5)};
{\ar@{-} (2,0.7);(2,0.3)};{\ar@{-} (3,0.7);(3,0.3)};(0.3,1)*[l]{I_i};
(1.7,0.5)*[l]{I_j};(0.9,-0.3)*{L_i};(2.3,-1.7)*{L_j};\endxy$
\caption{Illustrating the mapping between intervals and Ls for adjacent vertices
(left) and non-adjacent vertices (right).}
\label{fig:interval-L}
\end{figure}

$(\Longleftarrow)$ Let $\calL = \{L_1, \ldots, L_n\}$ be an anchored linear
L-system of $G$. Consider the interval representation $\calI=\{I_1, \ldots, I_n\}$ such
that $I_i = (c_{L_i},r_{L_i})$. The equivalence of $\calI$ and $\calL$ follows
similarly to ($\Longrightarrow$). 
\end{proof}

\begin{corollary}\label{cor:interval_in_mpt}
Interval graphs are a strict subclass of MPT graphs.
\end{corollary}
\begin{proof}
This follows from Proposition~\ref{pro:intervals_as_Ls} and the fact that the graph in Figure \ref{fig:L-rect-tri}
is an MPT graph but not an interval graph \cite{LB1962}.
\end{proof}

\section{Combinatorial Optimization Problems}
\label{sec:comp_opt}

In this section we will discuss the weighted independent set (WIS) problem, clique cover (CC) problem, and the coloring problem on MPT graphs. 
In particular, we will show the WIS problem can be solved in $\calO(n^3)$ time, the CC problem can be 2-approximated in quadratic time, the coloring problem is \NP-complete but can be $\log(n)$-approximated in linear time. 



Throughout this section we consider an MPT graph $G = (V,E)$ together with a
linear L-system $\calL = \{L_1, \ldots, L_n\}$ of $G$ where $i < j$ iff the
corner point of $L_i$ occurs to the left of $L_j$. Without loss of generality
we shall assume that the corner point of $L_i$ is $(i,-i)$ for each
$i \in \{1, \ldots, n\}$; i.e., $p_i=i$ in the corresponding MPT representation
and $L_i = (t_i, i, r_i)$. 

\subsection{Maximum Weight Independent Set}

The IS problem, even for the unweighted case, is known to be \NP-complete for: L-graphs, boxicity-2 graphs, and triangle intersection graphs since they contain the intersection graphs of vertical and horizontal line segments (also known as 2-DIR) and the problem is \NP-complete on 2-DIR \cite{KratochvilN1990}. 
Prior to \cite{KratochvilN1990}, the IS problem was known to be \NP-complete on boxicity-2 graphs \cite{FowlerPT1981,ImaiA1983}. 
However, for interval graphs, the WIS problem is known to be solvable in linear time from a superclass (e.g., \emph{chordal} graphs \cite{Frank1975-chordal-wis}) of interval graphs.
A graph is \emph{chordal} when it has no induced $k$-cycle for all $k \geq 4$.

Notice that an independent set in an MPT graph corresponds to a collection of disjoint L-shapes in a linear L-system. 
We use this equivalence to solve the WIS problem on a vertex-weighted MPT graph in polynomial time algorithm via dynamic programming. 
A close examination of our approach reveals its similarity to an algorithm for WIS on generalizations of interval graphs \cite{Lubiw1991}. The approach in \cite{Lubiw1991} also involves the use of dynamic programming with respect to certain intervals (which we simplify to \emph{dominant} L-shapes) and no specific time bound other than polynomial is claimed. 
However we believe our presentation is much clearer for the context of MPT graphs and it provides a direct time bound of $\calO(n^3)$. 
Also, there has been a recent $\calO(n^2)$ dynamic programming algorithm for this problem \cite{HittingRectangles2014} (this is based on \cite{Lubiw1991}), 
but here we believe that the simplicity of our approach provides insight into the structure of independent sets in MPT graphs and so we have included it.


We now discuss the key idea. Let $\calJ$ be a sub-collection of
disjoint L-shapes of $\calL$. We say that an L-shape $L_i$ is \emph{dominant} in
$\calJ$ if it contains the right-most point among the L-shapes in $\calJ$; i.e.,
$L_i \in \calJ$ and $r_i = \max_{L_j \in \calJ} r_j$.  Consider a dominant $L_i$ 
and some $L_j \in \calJ$ such that $j > i$. 
Notice that $L_j$ cannot contain any points to the right of
the line $x = r_j$ (since $L_i$ is dominant). Moreover,  $L_j$ must occur
strictly below the line $y = -i$ (since $L_j$'s corner point is below $L_i$'s
corner point). Similarly, for $L_{j'} \in \calJ$ with $j' < i$, $L_{j'}$ again
cannot contain any points to the right of the line $x = r_j$. Furthermore,
$L_{j'}$ is contained strictly above the line $y = -i$. Thus, for an L-shape
$L_i$, if $L_i$ is dominant in a sub-collection $\calJ$ of disjoint L-shapes of
$\calL$, then the L-shapes which belong to $\calJ$ and precede $L_i$ can be
chosen independently of the L-shapes which belong to $\calJ$ and follow $L_i$. 


The following notation is depicted in Figure \ref{fig:dominant}. 
For $a,b \in \{1, \ldots, n\}$ and $a \leq b$, let $\calL_{0,n+1} = \calL$ and 
$\calL_{a,b} = \calL_{0,b} \cap \calL_{a,n+1}$ where: 

\begin{itemize*}
\item $\calL_{0,b} = \{ L_i : 1 \leq i \leq b-1$, $r_i < r_b$, and $L_i \cap L_b
= \emptyset\}$; and 
\item $\calL_{a,n+1} = \{ L_i : a+1 \leq i \leq n$, $r_i < r_a$, and $L_i \cap
L_a = \emptyset\}$. 
\end{itemize*}


\begin{figure}[h]
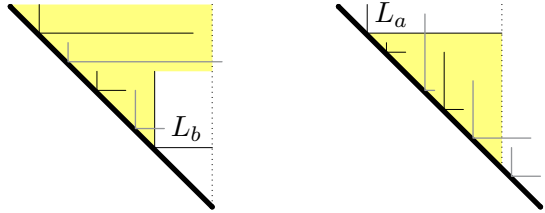
\centering\newxyColor{LightYellow}{1.0 1.0 0.5}{rgb}{}
$\xy/r3pc/:
@i@={(0,0),(1.5,-1.5),(1.5,-0.7),(2.1,-0.7),(2.1,0.0)},0*[LightYellow]\xypolyline{*};
{\ar@{-}@[|2pt] (0,-0.02);(2.1,-2.12)};
{\ar@{.} (2.1,-2.1);(2.1,0.0)};
@i@={(0.3,0),(0.3,-0.3),(1.9,-0.3)},0*\xypolyline{};
@i@={(0.6,-0.4),(0.6,-0.6),(2.2,-0.6)},0*[Gray][|0.5pt]\xypolyline{};
@i@={(0.9,-0.7),(0.9,-0.9),(1.2,-0.9)},0*\xypolyline{};
@i@={(1.5,-0.7),(1.5,-1.5),(2.1,-1.5)},0*\xypolyline{};
@i@={(1.3,-0.9),(1.3,-1.3),(1.6,-1.3)},0*[Gray][|0.5pt]\xypolyline{};
(1.8,-1.3)*[r]{L_b};
\endxy$
\qquad\qquad
$\xy/r3pc/:
@i@={(0.3,-0.3),(1.7,-1.7),(1.7,-0.3)},0*[LightYellow]\xypolyline{*};
{\ar@{-}@[|2pt] (0,-0.02);(2.1,-2.12)};
{\ar@{.} (1.7,-1.7);(1.7,0.0)};
@i@={(0.3,0),(0.3,-0.3),(1.7,-0.3)},0*\xypolyline{};
@i@={(0.5,-0.4),(0.5,-0.5),(0.7,-0.5)},0*\xypolyline{};
@i@={(0.9,-0.1),(0.9,-0.9),(1.0,-0.9)},0*[Gray][|0.5pt]\xypolyline{};
@i@={(1.1,-0.5),(1.1,-1.1),(1.3,-1.1)},0*\xypolyline{};
@i@={(1.4,-0.8),(1.4,-1.4),(2.0,-1.4)},0*[Gray][|0.5pt]\xypolyline{};
@i@={(1.8,-1.5),(1.8,-1.8),(2.1,-1.8)},0*[Gray][|0.5pt]\xypolyline{};
(0.5,-0.1)*[r]{L_a};
\endxy$
\caption{The L-shapes strictly contained in the shaded regions illustrate
$\calL_{0,b}$ (left) and $\calL_{a,n+1}$ (right).}
\label{fig:dominant}
\end{figure}

Let $opt[a,b]$ denote the maximum total weight of a collection of mutually
disjoint L-shapes in $\calL_{a,b}$. Notice that $opt[0,n+1]$ is the maximum
weight of an independent set in $G$. Furthermore, by the above discussion, we
have the following recurrence for $opt[a,b]$:

\par\hfill
$\displaystyle opt[a,b] = \max_{L_i \in \calL_{a,b}} \left(opt[a,i] + w(L_i) + opt[i,b] \right)$
\hfill\null\par\medskip

It is easy to see that the collection of sets $\{\calL_{a,b} : a,b \in
\{0,\ldots,n+1\}, a \leq b\}$ can be computed in $\calO(n^3)$ time (since each of
$\calL_{a,b}$ can be computed in $\calO(n)$ time). Moreover, the size of the table
$opt$ is $\calO(n^2)$, and the time to compute each entry is $\calO(n)$. Thus, we have
the following theorem. 

\begin{theorem}\label{thm:MWIS}
For a vertex weighted MPT graph with a given linear L-system, a maximum weight
independent set can be computed in $\calO(n^3)$ time. 
\end{theorem}

\subsection{Clique Cover}

The CC problem is known to be \NP-complete on boxicity-2 graphs (from unit square intersection graphs \cite{FowlerPT1981}), 
and L-graphs (from circle graphs \cite{KeilS2006}). However it is solvable in polynomial time on interval graphs and outerplanar graphs.

In this subsection we describe a polynomial time 2-approximation algorithm for the CC problem on MPT graphs. Our approach uses ideas similar to the algorithm for hitting set in \cite{Chepoi20131036}. 
From our algorithm we will see that the clique cover number $\gamma(G)$ is at most twice the independence number $\alpha(G)$ for any MPT graph $G$. 
Recently it has been observed that a hitting set for a linear rectangle-system can be 2-approximated
in polynomial time \cite{HittingRectangles2014}. Such a hitting set also provides a corresponding 
clique cover of the same size and their proof implies the $2\alpha(G)$ bound.  This proof uses a duality gap argument regarding the difference between the size of a MIS and and the size of a minimum hitting set and is quite different from our approach. Additionally, our approach is faster and simpler. 

Our algorithm begins with the linear L-system $\calL = \{L_1, ..., L_n\}$. Recall that $\calL$ is ordered according to the corner points of the L-shapes. From $\calL$ we greedily select an independent set $I$. We then build a partial clique cover of $G$ with one clique for each element of $I$. Finally, we consider the graph $H$ which remains after removing these cliques and observe that it is an interval graph. Since $H$ is an interval graph we can efficiently compute an optimal clique cover for it. This completes the overview of our algorithm. Notice that, since $H$ will be an interval graph (i.e., a perfect graph), $\gamma(H)=\alpha(H)$. Thus, the size of the clique cover that we produce is $|I| + \alpha(H) \leq 2\alpha(G)$. We now describe our algorithm in detail. 

First we construct the greedy independent set as follows. Let $I_1 = \{L_1\}$, and let $I_i = I_{i-1} \cup \{L_j\}$ such that $L_j$ does not intersect any L-shape in $I_{i-1}$ and $j$ is the smallest index satisfying this property. Let $I = \{L_{i_1}, ..., L_{i_k}\}$ be the maximal independent set constructed in this way such that $i_j < i_{j'}$ whenever $j < j'$. Since $I$ is an independent set in $G$, we can see that $k$ is at most the clique cover number of $G$. We will construct a partial clique cover using $I$ and show that the remaining graph $H$ will be an interval graph. 

To this end, consider the following disjoint sets of vertices. 
For each $j \in \{1, \ldots, k-1\}$, let 
$C_j = \{v_\ell : i_j \leq \ell < i_{j+1}$, and $r_\ell \geq i_{j+1}\}$.
First we claim that each such $C_j$ is a clique, and then we claim that removing all such $C_j$s from $G$ results in an interval graph $H$. 

\medskip
\noindent\textbf{Claim 1:} $C_j$ is a clique. 
\begin{proof}
Consider two vertices in $C_j$. Their corner points occur between the corners of $L_{i_j}$ and $L_{i_{j+1}}$, their top points occur above the corner of $L_{i_j}$ (otherwise one of them would be chosen into $I$ instead of $L_{i_{j+1}}$), and their right points occur to the right of the corner of $L_{i_{j+1}}$. Thus, they must intersect; i.e., $C_j$ is a clique.  
\end{proof}

\noindent\textbf{Claim 2:} $H = G \setminus (\bigcup_{j=1}^{k} C_j)$ is an interval graph. 
\begin{proof}
Consider $v_p$ in $H$ where $i_j \leq p < i_{j+1}$ and $1\leq j<k$. First, due to our construction of $I$, either $v_p = v_{i_j}$ or $v_p$ is a neighbor of some $v_{i_{j'}}$ where $i_{j'} \leq i_j$; i.e., the vertical segment of every such $v_p$ intersects the line $y = i_j$. 
Second, we know that the right-most point of $L_p$ is to the left of $L_{i_{j+1}}$ (since $v_p \notin C_j$).
This implies that every neighbor $v_q$ of $v_p$ in $H$ has $i_j \leq q < i_{j+1}$.
Thus, $H$ induced on its vertices between $v_{i_j}$ and $v_{i_{j+1}}$ is an interval graph (since it has an anchored linear L-system anchored at $i_j$) and is a disjoint union of connected components of $H$. 

The same argument applies to vertices $v_p$ with $i_k\leq p$.
This show that $H$ is the disjoint union of interval graphs; i.e., $H$ itself is an interval graph.
\end{proof}

Notice that the greedy independent set as well as the cliques $C_j$ are easily generated 
in linear time. Moreover, the CC problem on interval graphs can be solved in linear time \cite{HT1991-CC-interval}. 
This leads to the main theorem of this subsection. 

\begin{theorem}\label{thm:clique-cover}
For an MPT graph $G$ the clique cover number is at most twice the independence number. Also, when a linear L-system is given as input, the clique cover $G$ can be 2-approximated in $\calO(n+m)$ time. 
\end{theorem}

\subsection{Coloring}

The coloring problem is known to be \NP-complete on L-graphs (since circle graphs, also known as interval overlap graphs, 
are contained in L-graphs \cite{ACGLLS2012-VPG} and coloring circle graphs is
\NP-complete \cite{GJMP1980-circarc+circle}), on boxicity-2 graphs
\cite{ImaiA1983}, and on triangle intersection graphs (since they include planar
graphs \cite{FraysseixMR1994-triangle} and coloring is \NP-complete on planar
graphs \cite{GJ1979}). On the other hand, the coloring problem can be solved in
linear time on interval graphs \cite{G-AlgGraphTheory2004} and outerplanar graphs 
\cite{PrSys1986-outerplanar-col}. 

In this section we will demonstrate that it is \NP-complete to determine the
chromatic number for MPT graphs, but it can be $\log(n)$-approximated in polynomial time. 
We will use $\chi(G)$ to denote the chromatic number of $G$. 

Prior to proving the hardness result we observe that $\chi(G)$ can be $\log(n)$-approximated using known techniques. 
For any boxicity-2 graph $G$, the relationship between the $\chi(G)$ and $\omega(G)$ (the clique number) has been well-studied. 
The best results regarding this relationship are given in \cite{Chalermsook2011}. The relevant result for MPT graphs is as follows. For a boxicity-2 graph $G$ with a rectangle system such that no rectangle contains another, $\chi(G)$ is $\calO(\omega(G)\log(\omega(G)))$ and this $\log(n)$-approximation of $\chi(G)$ can be computed in polynomial time. It is easy to see from our characterization of MPT graphs as linear boxicity-2 graphs, that this result applies directly to MPT graphs. Thus, the chromatic number of MPT graphs can be $\log(n)$-approximated in polynomial time. 

We now turn to the hardness of coloring for MPT graphs. 
To do this we transform the hardness of coloring of circular-arc
graphs to this class. \emph{Circular-arc} graphs are the intersection graphs
of arcs of a circle. 
 Determining a minimum coloring of a circular-arc graph is
known to be NP-hard \cite{GJMP1980-circarc+circle}; i.e., it is 
NP-complete to determine whether a circular arc graph is $k$ colorable when $k$
is part of the input. 

\begin{theorem}\label{thm:chromatic_number}
It is NP-complete to determine the chromatic number for MPT graphs. 
\end{theorem} 
\begin{proof} 
Consider a circular-arc graph $G = (V,E)$. We use $n$ and $m$ to denote $|V|$
and $|E|$ respectively. Now, for any $k > 2$, we will construct an MPT graph $G'
= (V',E')$ such that: $|V'|= O(n)$, $|E'|=O(n^2)$, and $\chi(G) \leq k$
iff $\chi(G') \leq k$. Moreover, $G'$ is easily constructed in $\calO(n^2)$ time. An
example of this construction is depicted in Figure \ref{fig:coloring}. 
The basic idea is that we ``cut'' the circular-arc representation at an arbitrary point $p$. This point corresponds to a clique and we split every vertex crossing this point into two vertices so that the result is an interval graph. This interval graph has an anchored linear L-system to which we add a clique consisting of $k$ vertices. This clique will ensure that in any coloring of this constructed graph, the two copies of every split vertex have the same color. We now present the formal proof.

\begin{figure}[t!]
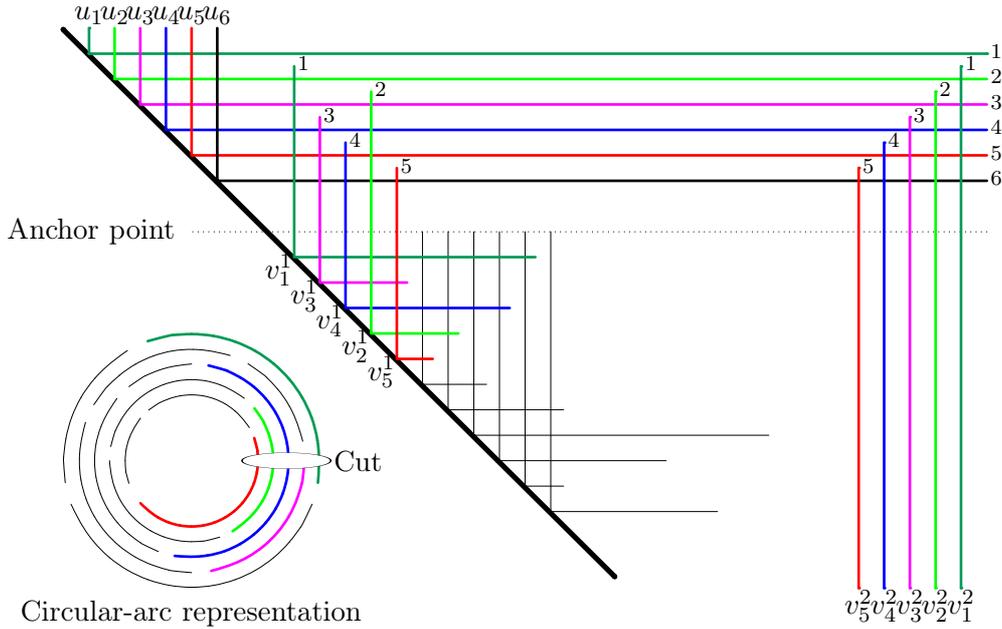
\centering
$\xy/r0.8pc/:
{\ar@{-}@[|2pt] (0,-0.05);(21.5,-21.55)};
{\ar@{-}@[|1pt]@[ForestGreen] (1,0);(1,-1)};
{\ar@{-}@[|1pt]@[ForestGreen] (1,-1);(36,-1)};
@i@={(2,0),(2,-2),(36,-2)},0*[|1pt][green]\xypolyline{};
@i@={(3,0),(3,-3),(36,-3)},0*[|1pt][magenta]\xypolyline{};
@i@={(4,0),(4,-4),(36,-4)},0*[|1pt][blue]\xypolyline{};;
@i@={(5,0),(5,-5),(36,-5)},0*[|1pt][red]\xypolyline{};
@i@={(6,0),(6,-6),(36,-6)},0*[|1pt][black]\xypolyline{};
{\ar@{.} (5,-8);(36,-8)};(5,-8)*[l]{\mbox{Anchor point~~~}};
@i@={(14,-8),(14,-14),(16.5,-14)},0*\xypolyline{};
@i@={(15,-8),(15,-15),(19.5,-15)},0*\xypolyline{};
@i@={(16,-8),(16,-16),(27.5,-16)},0*\xypolyline{};
@i@={(17,-8),(17,-17),(23.5,-17)},0*\xypolyline{};
@i@={(18,-8),(18,-18),(19.5,-18)},0*\xypolyline{};
@i@={(19,-8),(19,-19),(25.5,-19)},0*\xypolyline{};
@i@={(9,-1.5),(9,-9),(18.4,-9)},0*[|1pt][ForestGreen]\xypolyline{};
@i@={(10,-3.5),(10,-10),(13.4,-10)},0*[|1pt][magenta]\xypolyline{};
@i@={(11,-4.5),(11,-11),(17.4,-11)},0*[|1pt][blue]\xypolyline{};
@i@={(12,-2.5),(12,-12),(15.4,-12)},0*[|1pt][green]\xypolyline{};
@i@={(13,-5.5),(13,-13),(14.4,-13)},0*[|1pt][red]\xypolyline{};
{\ar@{-}@[|1pt]@[red] (31,-5.5);(31,-22)};
{\ar@{-}@[|1pt]@[blue] (32,-4.5);(32,-22)};
{\ar@{-}@[|1pt]@[magenta] (33,-3.5);(33,-22)};
{\ar@{-}@[|1pt]@[green] (34,-2.5);(34,-22)};
{\ar@{-}@[|1pt]@[ForestGreen] (35,-1.5);(35,-22)};
(1,0.5)*{u_1};(2,0.5)*{u_2};(3,0.5)*{u_3};(4,0.5)*{u_4};(5,0.5)*{u_5};(6,0.5)*{u_6};
(8.4,-9.5)*{v^1_1};(9.4,-10.5)*{v^1_3};(10.4,-11.5)*{v^1_4};(11.4,-12.5)*{v^1_2};(12.4,-13.5)*{v^1_5};
(31,-22.7)*{v^2_5};(32,-22.7)*{v^2_4};(33,-22.7)*{v^2_3};(34,-22.7)*{v^2_2};(35,-22.7)*{v^2_1};
(9.4,-1.4)*{_1};(10.4,-3.4)*{_3};(11.4,-4.4)*{_4};(12.4,-2.4)*{_2};(13.4,-5.4)*{_5};
(36.4,-0.9)*{_1};(36.4,-1.9)*{_2};(36.4,-2.9)*{_3};(36.4,-3.9)*{_4};(36.4,-4.9)*{_5};(36.4,-5.9)*{_6};
(31.4,-5.4)*{_5};(32.4,-4.4)*{_4};(33.4,-3.4)*{_3};(34.4,-2.4)*{_2};(35.4,-1.4)*{_1};
{(5,-17),{\xypolygon36"A"{~:{(2.6,0):}~>{}~={0}{}}},
{\xypolygon36"B"{~:{(3.2,0):}~>{}~={0}{}}},
{\xypolygon36"C"{~:{(3.8,0):}~>{}~={0}{}}},
{\xypolygon36"D"{~:{(4.4,0):}~>{}~={0}{}}},
{\xypolygon36"E"{~:{(5,0):}~>{}~={0}{}}}};
{"A4";"A14"**\crv{"A5"&"A6"&"A7"&"A8"&"A9"&"A10"&"A11"&"A12"&"A13"&}};
{"A15";"A21"**\crv{"A16"&"A17"&"A18"&"A19"&"A20"&"A20"&}};
{"A23";"A3"**[red][|1pt]\crv{"A24"&"A25"&"A26"&"A27"&"A28"&"A29"&"A30"&"A31"&"A32"&"A33"&"A34"&"A35"&"A36"&"A1"&"A2"&}};
{"B6";"B18"**\crv{"B7"&"B8"&"B9"&"B10"&"B11"&"B12"&"B13"&"B14"&"B15"&"B16"&"B17"&}};
{"B19";"B30"**\crv{"B20"&"B21"&"B22"&"B23"&"B24"&"B25"&"B26"&"B27"&"B28"&"B29"&}};
{"B31";"B5"**[green][|1pt]\crv{"B32"&"B33"&"B34"&"B35"&"B36"&"B1"&"B2"&"B3"&"B4"&}};
{"C10";"C14"**\crv{"C11"&"C12"&"C13"&}};{"C15";"C26"**\crv{"C16"&"C17"&"C18"&"C19"&"C20"&"C21"&"C22"&"C23"&"C24"&"C25"&}};
{"C27";"C9"**[blue][|1pt]\crv{"C28"&"C29"&"C30"&"C31"&"C32"&"C33"&"C34"&"C35"&"C36"&"C1"&"C2"&"C3"&"C4"&"C5"&"C6"&"C7"&"C8"&}};
{"D2";"D7"**\crv{"D3"&"D4"&"D5"&"D6"&}};{"D8";"D16"**\crv{"D9"&"D10"&"D11"&"D12"&"D13"&"D14"&"D15"&}};
{"D17";"D28"**\crv{"D18"&"D19"&"D20"&"D21"&"D22"&"D23"&"D24"&"D25"&"D26"&"D27"&}};
{"D29";"D1"**[magenta][|1pt]\crv{"D30"&"D31"&"D32"&"D33"&"D34"&"D35"&"D36"&}};
{"E13";"E20"**\crv{"E14"&"E15"&"E16"&"E17"&"E18"&"E19"&}};
{"E21";"E35"**\crv{"E22"&"E23"&"E24"&"E25"&"E26"&"E27"&"E28"&"E29"&"E30"&"E31"&"E32"&"E33"&"E34"&}};
{"E36";"E12"**[ForestGreen][|1pt]\crv{"E1"&"E2"&"E3"&"E4"&"E5"&"E6"&"E7"&"E8"&"E9"&"E10"&"E11"&}};
(8.7,-17)*[c][o]+[F**:white]{\hspace*{2.5em}};(11.5,-17)*{\mbox{Cut}};
(5,-23)*{\mbox{Circular-arc representation}};\endxy$
\caption{Sample construction from the proof of Theorem 
\ref{thm:chromatic_number} where the ``cut'' contains $5$ vertices and 
$k=6$.}
\label{fig:coloring}
\end{figure}


Consider an arbitrary circular-arc representation $\mathcal{A}$ of $G$
(such a representation can be constructed in $\calO(n+m)$ time
\cite{McConnell2003-circulararc}). Let $p$ be a fixed point on the circle of
$\mathcal{A}$ and let $\mathcal{A}_p = \{A_1, \ldots, A_\ell\}$ be the arcs of
$\mathcal{A}$ that include $p$. The vertices $\{v_1, \ldots, v_\ell\}$
corresponding to $\mathcal{A}_p$ form a clique in $G$ (since the arcs all share
the point $p$). Hence, if no arcs pass through the point $p$, then $G$ is an
interval graph; i.e., $G$ is an MPT graph and so we can let $G'=G$ and we are
done.  Similarly, if $\ell > k$, then $\chi(G) > k$ and we are done; i.e., we
simply let $G'$ be a clique on $\ell$ vertices. Thus we may assume $1 \leq \ell \leq k$. 

We now form an interval graph $H$ from $G$ by ``cutting'' the circular-arc
representation $\mathcal{A}$ at the point $p$. Formally, for some small enough
$\epsilon > 0$ and each $i \in \{1, \ldots, \ell\}$, we replace the arc $A_i =
(s_i,e_i)$ with two arcs $A^1_i = (s_i, p - \epsilon)$, and $A^2_i = (p+
\epsilon, e_i)$ and consider $H$ as the resulting intersection graph. In
particular, each vertex $v_i$ is replaced by two vertices $v^1_i$ and $v^2_i$
corresponding to the arcs $A^1_i$ and $A^2_i$ respectively. Notice that $|V(H)|
= n + \ell$ and $|E(H)| \leq 2m$.  Since there are no arcs
passing through the point $p$ in this circular-arc representation of $H$, the
graph $H$ is an interval graph. Thus, by Proposition~\ref{pro:intervals_as_Ls}, $H$
has an anchored linear L-system. 

Finally, we add a clique of size $k$ to $H$ so that the result is an MPT graph
$G'$ and in any $k$-coloring of $G'$, the vertices $v^1_i$ and $v^2_i$ must be assigned the
same color. To this end, we define $G' = (V',E')$ as follows:\smallskip

$V' = V(H) \cup \{u_1, \ldots, u_{k}\}$,\smallskip

$E' = E(H) \cup \big\{u_t v^j_i : j \in\{1,2\}, i \in \{1, \ldots, \ell\}, t \in
\{i+1, \ldots, k\}\big\} \cup \big\{u_iu_j : i,j \in \{1, \ldots, k\}, i \neq j\big\}$. 
\medskip

\noindent We show that $G'$ has a $k$-coloring iff $\chi(G) \leq
k$.\smallskip

\noindent$\Longrightarrow$ Notice that the vertices $v^1_1$ and $v^2_1$ are adjacent to
the same clique of size $k-1$ in $G'$. Thus, in any $k$-coloring of $G'$,
$v^1_1$ and $v^2_1$ must be assigned the same color. Inductively, it is easy to
see that $v^1_i$ and $v^2_i$ must also receive the same color in any
$k$-coloring of $G'$. Specifically, $u_i$, $v^1_i$, and $v^2_i$ will receive
the same color for every $i \in \{1, \ldots, \ell\}$.  Thus, any $k$-coloring of
$G'$ provides a $k$-coloring of $G$.

\noindent$\Longleftarrow$ We can extend any
$k$-coloring $f : V(G) \rightarrow \{1, \ldots, k\}$ of $G$ to a $k$-coloring $f'
: V(G') \rightarrow \{1, \ldots, k\}$ of $G'$ as follows. For every $v \in V(G)
\setminus \{v_1, \ldots, v_\ell\}$, set $f'(v) = f(v)$. For each $i \in \{1,
\ldots, \ell\}$, set $f'(u_i)=f'(v^1_i) = f'(v^2_i) = f(v_i)$, and then choose $f'(u_{\ell+1}), \ldots,
f'(u_{k})$ so that $\{f'(u_{\ell+1}), \ldots, f'(u_{k})\} = \{1,\ldots,k\} \setminus
\{f(v_1), \ldots, f(v_{\ell})\}$. It is easy to see that $f'$ is a
$k$-coloring of $G'$. This completes the proof of the claim.\smallskip

All that remains is to show that $G'$ has the appropriate size and that it is an
MPT graph. Notice that $|V(G')| = n + \ell + k \leq 3n$ and $|E(G')| \leq 2m +
{k \choose 2} + (k-\ell)*2\ell + \sum_{t = 1}^{\ell-1}2t \leq
3n^2$. Thus, $G'$ has the appropriate size. Furthermore, we can construct an
MPT representation of $G'$ by starting from an anchored linear L-system of $H$
and adding L-shapes for the new clique ``above'' this anchored linear L-system
(see Figure \ref{fig:coloring}). Thus, $G'$ is an MPT graph. 

From the above construction we can see that determining the chromatic number for
MPT graphs is NP-hard, since it is NP-hard to determine the chromatic
number for circular-arc graphs. 
\end{proof}


This leaves open the $k$-coloring problem for fixed $k\geq 3$.  In particular,
note that in the above construction it was necessary that the number of colors
$k$ was part of the input, since for fixed $k$, the $k$-coloring problem is
solvable in polynomial time on circular-arc graphs
\cite{GJMP1980-circarc+circle}.

\section{Other Characterizations}

In this section we characterize MPT graphs by linear vertex orders, the intersection of interval graphs, and as a restricted class of segment graphs.  

\subsection{Vertex Ordering}

Several well known graph classes have been characterized by special linear
orders on their vertices; e.g., interval graphs (see Definition
\ref{def:i-order} and Theorem \ref{thm:i-order}), unit interval graphs
\cite{Roberts1968-unitinterval}, chordal graphs \cite{Dirac1961-peo}, and
co-comparability graphs \cite{KratschS1993-cocomp}. In this section we
characterize MPT graphs as graphs with MPT-orders (see Definition
\ref{def:mpt-order} and Theorem \ref{thm:mpt-order}). 
This characterization is also stated in \cite{c-And2013}.
We then use this ordering to show that a graph is an MPT graph iff it is the 
intersection of two ``special'' interval graphs 
(see Theorem \ref{thm:2interval-MPT}).

\begin{definition}\label{def:i-order}
An \emph{I-order} of a graph $G$ with vertices $v_1, \ldots, v_n$ is an ordering
$v_1 < v_2 < \cdots < v_n$ such that: for every $u<v<w$, if $uw \in E(G)$, then $uv
\in E(G)$. 
\end{definition}

\begin{theorem}\label{thm:i-order}
\cite{Olariu1991-iorder,RR1988-iorder,Raychaudhuri1987-iorder} $G$ is an
interval graph iff $G$ has an I-order. Moreover for any interval representation
$\calI$ of a graph $G$, ordering the vertices of $G$ by the left end-points
of their intervals results in an I-order of $G$.
\end{theorem}

\begin{definition}\label{def:mpt-order}
An \emph{MPT-order} of a graph $G$ with vertices $v_1, \ldots, v_n$ is an ordering
$v_1 < v_2 < \cdots < v_n$ such that: for every $u<v<w<x$, if $uw,vx \in E(G)$,
then $vw \in E(G)$. 
\end{definition}

Notice that MPT-order is a generalization of I-order. In particular, let
$\sigma$ be an I-order of a graph $G$. Now suppose we have $u,v,w,x \in V(G)$
such that $u<_\sigma v<_\sigma w<_\sigma x$ and $uw,vx \in E(G)$. Since $\sigma$
is an I-order with $v<_\sigma w <_\sigma x$ and $vx \in E(G)$, the edge $vw$ is
forced. Thus, $\sigma$ is an MPT-order; i.e., every I-order is also an MPT-order. 
We now prove that MPT graphs are characterized as the graphs with MPT-orders. 

\begin{theorem}\label{thm:mpt-order}
$G=(V,E)$ is an MPT graph iff $G$ has an MPT-order (i.e., the vertices of $G$
can be ordered by $<$ so that for every $u,v,w,x \in V$, if $u<v<w<x$ and $uw,
vx \in E$, then $vw \in E$).  
\end{theorem}
\begin{proof}
($\Longrightarrow$) Let $\{(s_v,p_v,e_v): v \in V\}$ be an MPT representation of
$G$. Order the vertices of $G$ such that vertex $v$ comes before vertex $u$ if
$p_{v}\leq p_{u}$. Now, consider any four distinct vertices $u,v,w,x$ where
$u<v<w<x$ and $uw,vx \in E$. Then, it is easy to realize that, due to the
considered ordering, it holds that $s_{w} \leq p_{u} \leq p_{v}$ and $e_{v} \geq
p_{x} \geq p_{w}$, which implies $vw \in E$. 	

($\Longleftarrow$) Let $G=(V,E)$ be a graph with ordered vertex set $V = \{v_1,
\ldots, v_n\}$ such that for any $i,j,k,\ell \in \{1,\ldots,n\}$, if $i < j < k <
\ell$ and $v_{i}v_k,v_{j}v_\ell \in E$ then $v_{k}v_j \in E$ (i.e., $v_1 <
\cdots < v_n$ is an MPT-order). We now construct an MPT representation of $G$ based on
this ordering. For each $i \in \{1, \ldots, n\}$, let: \begin{itemize*} \item $s_i
= \min\{i,j\}$ where $j$ is the \textbf{smallest} index such that $v_jv_i$ is an
edge in $G$.  \item $p_i = i$ \item $e_i = \max\{i,j\}$ where $j$ is the
\textbf{largest} index such that $v_iv_j$ is an edge in $G$.  \end{itemize*}
Clearly $\calI = \{(s_i,p_i,e_i) : i \in \{1, \ldots, n\}\}$ is an MPT
representation in which every edge of $G$ is captured. Now we need to
demonstrate that this representation does not include any edges which are not
edges of $G$. Suppose that for some $j,k \in \{1,\ldots,n\}$, $j<k$, $v_jv_k \notin
E$ but $s_k \leq p_j$ and $e_j \geq p_k$. Since $s_k \leq p_j$ there must be
$v_i$ with $i < j$ such that $v_iv_k \in E$. Similarly, there must be $v_\ell$
with $\ell > k$ such that $v_jv_\ell \in E$. However, we now have $i < j < k <
\ell$ with $v_iv_k, v_jv_\ell \in E$ but $v_jv_k \notin E$; i.e., a
contradiction to the vertex order. Thus $\calI$ is an MPT representation of $G$. 
\end{proof}

Notice that, since every I-order is an MPT-order and every graph with an
MPT-order is an MPT graph, we have an alternate proof of Corollary
\ref{cor:interval_in_mpt}; i.e., that every interval graph is an MPT graph.
Also, since the order of vertices in an MPT-order corresponds to the order of
the points in an MPT representation, they also correspond to the order of the
corner points in a linear L-system of an MPT graph. 

We conclude this section by further characterizing MPT graphs as the
intersection of two related interval graphs.

\begin{theorem}\label{thm:2interval-MPT}
$G=(V,E)$ is an MPT graph with MPT-order $\sigma = (v_1 < \ldots < v_n)$ iff there
are interval graphs $H_1 = (V,E_1)$ and $H_2 = (V,E_2)$ such that $E = E_1 \cap
E_2$, $\sigma$ is an I-order of $H_1$, and the reverse of $\sigma$ (i.e., $v_n <
\cdots < v_1$) is an I-order of $H_2$. 
\end{theorem}
\begin{proof}
($\Longrightarrow$) Let $\sigma=v_1<\cdots<v_n$ be an MPT-order of $G$, and let
$\calL$ be the linear L-system of $G$ constructed in the proof of Theorem
\ref{thm:mpt-order} using this order; i.e., $c_{v_i} < c_{v_j}$ iff $i<j$.
Construct an anchored linear L-system $\calL_1$ by extending the horizontal
segment of every L-shape in $\calL$ to the right beyond the corner of the
right-most L-shape in $\calL$. Similarly, construct an anchored linear L-system
$\calL_2$ by extending the vertical segment of every $L$ in $\calL$ so that it
reaches above the corner of the left-most L-shape in $\calL$. By Proposition~\ref{pro:intervals_as_Ls}, each of $\calL_1$ and $\calL_2$ corresponds to
an interval representation.  Let $H_1 = (V,E_1)$ and $H_2 = (V,E_2)$ be the
interval graphs specified by $\calL_1$ and $\calL_2$. Notice that, by Theorem
\ref{thm:i-order}, $\sigma$ is an I-order of $H_1$ and the reverse of $\sigma$
is an I-order of $H_2$. Thus, we just need to ensure that $E_1 \cap E_2 = E$.
Clearly $E \subseteq E_1 \cap E_2$ from our construction of $\calL_1$ and
$\calL_2$. Moreover there can be no edge in both $E_1$ and $E_2$ which is not in
$E$ simply due to how these ``extra'' edges come into existence (see Figure
\ref{fig:2interval-MPT}).

\begin{figure}[ht]
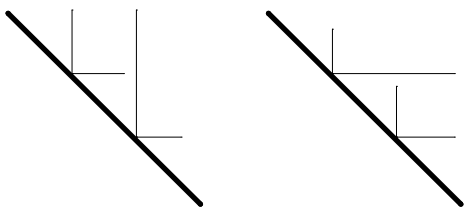
\centering
$\xy/r2pc/:{\ar@{-}@[|2pt] (0,-0.05);(3,-3.05)};{\ar@{-} (1,0);(1,-1)};
{\ar@{-} (1,-1);(1.8,-1)};{\ar@{-} (2,0);(2,-2)};{\ar@{-} (2,-2);(2.7,-2)};\endxy$
\qquad
$\xy/r2pc/:{\ar@{-}@[|2pt] (0,-0.05);(3,-3.05)};{\ar@{-} (1,-0.3);(1,-1)};
{\ar@{-} (1,-1);(2.9,-1)};{\ar@{-} (2,-1.2);(2,-2)};{\ar@{-} (2,-2);(2.9,-2)};
\endxy$
\caption{(left) $L_u,L_v$ such that $uv \in E_1 \setminus E$. (right) $L_u,L_v$
such that $uv \in E_2 \setminus E$.}
\label{fig:2interval-MPT}
\end{figure}

($\Longleftarrow$) Let $H_1 = (V,E_1)$ and $H_2 = (V,E_2)$ such that $V = \{v_1,
\ldots v_n\}$, $\sigma = (v_1 < \cdots < v_n)$ is an I-order of $H_1$, and the
reverse of $\sigma$ is an I-order of $H_2$. We now claim that $\sigma$ is an
MPT-order of $G = (V, E_1 \cap E_2)$. Consider $1\leq i<j<k<\ell \leq n$ where
$v_iv_k,v_jv_\ell \in E_1 \cap E_2$. Notice that $v_jv_k \in E_1$ since $\sigma$
is an I-order of $H_1$ and $v_jv_\ell \in E_1$. Similarly, $v_jv_k \in E_2$
since the reverse of $\sigma$ is an I-order of $H_2$ and $v_iv_\ell \in E_2$.
Thus, $v_jv_k \in E_1 \cap E_2$ as needed. 
\end{proof}

\subsection{Cyclic Segment Graphs}

In this section we characterize MPT graphs as intersection graphs of line segments from a cyclic line arrangement. A \emph{line arrangement} is simply a collection of lines in the plane (see \cite{felsner2004geometric} for more on line arrangements). A line arrangement $\calA$ is \emph{cyclic} when there is a convex function $f$ (e.g., a parabola) such that every line in $\calA$ is tangent to $f$. We define \emph{cyclic segment graphs} as the intersection graphs of line segments where the underlying line arrangement is cyclic with respect to some function $f$ and each segment contains a point on $f$. 
In the following theorem we prove that cyclic segment graphs are precisely MPT graphs. This follows easily from our characterization of MPT graphs via MPT-orders (see Theorem \ref{thm:mpt-order}). 

\begin{theorem}\label{thm:MPT=cyclic-seg}
MPT graphs are precisely cyclic segment graphs.
\end{theorem}
\begin{proof}
Let $\sigma = (v_1 < \ldots < v_n)$ be an MPT-order of an MPT graph $G$. We will construct a cyclic segment representation of $G$ by mapping each vertex to a segment of a line tangential to the parabola $y=x^2$. 
First, we assign each $v_i$ the tangent line $\ell_i$ of the parabola for the point $(i,i^2)$. 
Now, to choose the segment of $\ell_i$ for the vertex $v_i$ we consider the left-most and right-most vertices, say $v_{i_{\min}}$ and $v_{i_{\max}}$, from $N(v_i) \cup \{v_i\}$. In particular, we let the segment $S_i$ for $v_i$ be defined as the segment of $\ell_i$ starting from $\ell_{i_{\min}}$ and ending on $\ell_{i_{\max}}$. Note, if $i = i_{\min}$ ($i = i_{\max}$) then we simply use the point $(i,i^2)$ as the starting (ending) point of the segment $\ell_i$. 
Clearly each $S_i$ passes through the point $(i,i^2)$. Thus, we have constructed a valid cyclic segment representation. Consider an edge $v_iv_j$ of $G$ with $i<j$. From our construction, $S_i$ passes through the line $\ell_j$ in order to reach $S_{i_{\max}}$. Similarly, $S_j$ passes through the line $\ell_i$ in order to reach $S_{j_{\min}}$. Thus, $S_i$ and $S_j$ intersect. Now, suppose that $S_i$ and $S_j$ intersect ($i < j$), but $v_iv_j$ is not an edge of $G$. In order for these segments to intersect, each must need to ``reach over'' the other. In particular, this means that there is $v_p$ and $v_q$ such that $p < i < j < q$, and $v_iv_q$ and $v_jv_p$ are edges in $G$; i.e., this violates the MPT-order. Therefore, every MPT graph has a cyclic segment representation. 

To construct an MPT-order from a cyclic segment representation one simply uses the order of the tangent points and the proof follows similarly to the above.
\end{proof}

With our characterization established we note that it may be interesting to consider generalizations in this context. In particular, one might consider intersection graphs of line segments which are tangent to convex bodies or unimodal functions in $\mathbb{R}^2$. 

\section{Non-MPT graphs and More Subclasses of MPT graphs}
In this section we observe two additional strict subclasses of MPT graphs (namely, outerplanar graphs and 2D ray graphs). We further observe infinite families of graphs which are not MPT graphs. 

\subsection{Outerplanar graphs}

In this section we consider outerplanar graphs as a restricted form of MPT graphs. In particular, we consider linear L-contact-systems and demonstrate that the graphs of these contact systems are precisely outerplanar graphs.
It has been independently stated that outerplanar graphs are a subclass of MPT graphs \cite{c-And2013}. Their proof is completely different from ours and does not provide the characterization we have observed.

A graph is \emph{outerplanar} if it has a crossing-free embedding in the plane such that all vertices are on the same face. Moreover, an outerplanar graph is said to be maximal when it is not a proper subgraph of any outerplanar graph with the same number of vertices. We will demonstrate that outerplanar graphs are precisely the \emph{linear contact L-graphs} (see Definition \ref{def:contact-L} and Theorem \ref{thm:outerplanar}). For more information on contact L-graphs see \cite{ChUe2013,ChKoUe2013,DBLP:conf/soda/KobourovUV13}.

\begin{definition}\label{def:contact-L}
A graph $G$ is a \emph{linear L contact graph} when it has a linear L-system $\calL$ such that no two L-shapes ``cross-over'' each other; i.e., for L-shapes $L_u = (t_u,c_u,r_u), L_v = (t_v,c_v,r_v)$, if $L_u \cap L_v \neq \emptyset$ and $c_u < c_v$, then either $c_u = t_v$ or $r_u = c_v$. 
Moreover, we say such an L contact system is \emph{equilateral} when, for each L-shape, the vertical and horizontal segments have the same length. 
\end{definition} 

We will use the following characterization of maximal outerplanar graphs related to 2-trees (which follows easily from \cite{Sandr-1979-outerplanar}). A \emph{2-tree} is a graph that can be constructed by starting from an edge and iteratively adding vertices with exactly two adjacent neighbors. Semi-squares will also play a role throughout this section. A \emph{semi-square} is a right-triangle whose vertical and horizontal sides are the same length (e.g., the lower-left ``half'' of a square). Notice that, there are four types of semi-squares depending on the choice of corner: lower-left (ll), lower-right (lr), upper-left (ul), and upper-right (ur). It is known that max tolerance graphs are precisely semi-square intersection graphs where every semi-square has the same type (e.g., max tolerance graphs are precisely the ll-semi-square intersection graphs) \cite{Kaufmann2006-MaxTol-SODA}.

\begin{theorem}[\cite{Sandr-1979-outerplanar}]
Let $G$ be a maximal outerplanar graph. For any edge $v_1v_2$ of the outerface of $G$, the vertices of $G$ can be ordered $v_1, \ldots, v_n$ such that $v_i$ ($2 < i \leq n$) has exactly two neighbors, $u$ and $v$, in $G_{i-1} = G[\{v_1, \ldots, v_{i-1}]$ and $uv$ is and edge of $G$. We refer to such an order as an \emph{outerplanar-order}.
\end{theorem}


\begin{theorem}\label{thm:outerplanar}
Every maximal outerplanar graph $G$ is a linear equilateral-L contact graph. 
\end{theorem}
\begin{proof}
Consider an outerplanar order $v_1,v_2, \ldots, v_n$ of $G$. We iteratively build the linear equilateral-L contact system as follows. 
Let $L_{v_1} = (-1,0,1)$ and $L_{v_2} = (0,1,2)$ be the L-shapes for $v_1$ and $v_2$ respectively. Clearly $L_{v_1}$ and $L_{v_2}$ contact each other at the point $(1,0)$, both terminate at this point, are equilateral, and their corner points lie on the line $y=-x$. Moreover, the ur-semi-square defined by the points $(0,0), (1,0), (1,-1)$ is: 
\begin{itemize}
\item
empty (i.e., it is internally disjoint from all the $L$'s we have so far), 
\item
its diagonal is a segment of the line $y=-x$, and 
\item the point $(1,0)$ is the point of contact between $L_{v_1}$ and $L_{v_2}$. 
\end{itemize}
Now assume that we have a linear equilateral-L contact system $\calL_{i-1}$ for $v_1, \ldots, v_{i-1}$ such that every edge $uv$ on the outerface of $G_{i-1}$ corresponds to an empty ur-semi-square as in the base case. We extend this representation to a representation of $G_i$ as follows. From the outerplanar ordering, the vertex $v_i$ is adjacent in $G_{i-1}$ to precisely one pair $u,v$ such that $uv$ is an edge on the outerface of $G_{i-1}$. Thus, we have an empty ur-semi-square $(x,-x),(x+d,-x),(x+d,-x-d)$ where the point $(x+d,-x)$ is the contact point of $L_u$ and $L_v$. Consider $\calL_i = \calL_{i-1} \cup \{L_{v_i}\}$ where $L_{v_i} = (x,x+d/2,x+d)$. Without loss of generality $L_{v_i}$ contacts $L_u$ at the point $(x+d/2,-x)$ and $L_v$ at the point $(x+d,-x-d/2)$. Moreover, these new contact points form the appropriate empty and disjoint semi-squares as needed. Finally, since the semi-square corresponding to $uv$ was empty before inserting $L_{v_i}$, the L-shape $L_{v_i}$ does not intersect any other L-shapes. Thus, $\calL_i$ is a linear equilateral L contact system as needed. 
\end{proof}

Similarly to how linear L-graphs are equivalent to linear boxicity-2 graphs and linear right-triangle graphs, we have the following corollary regarding linear contact graphs. 

\begin{corollary}\label{cor:contacts}
The following graph classes are equivalent: outerplanar, linear L contact, linear equilateral-L contact, linear ll-semi-square contact, linear square contact. 
\end{corollary}
\begin{proof}
Since maximal outerplanar graphs are linear equilateral L contact graphs (by Theorem \ref{thm:outerplanar}), all outerplanar graphs are linear equilateral L contact graphs. In particular, one may simply adjust an equilateral L a by small amount to remove any individual contact with another L such that no other contact is altered. 

Moreover, given any of the contact representations listed, one can easily construct an outerplanar drawing of the graph. In particular, each vertex $v$ is located at its corresponding corner point on the line $y=-x$ and the edges $uv$ are drawn by tracing $L_u$ and $L_v$ to the corresponding contact point. Clearly all vertices lie on the outside of such a drawing and this can be done so that no edges intersect. 
\end{proof}

\subsection{2D Ray Graphs}

A graph is a \emph{2D ray} graph when it is an intersection graph of rays in the
plane where the rays have at most two directions and parallel rays do not
intersect (i.e., this is a bipartite graph class). For more information on this
graph class (including its relationship to many well-known graph classes) see \cite{Kostochka1998,OrthogonalRay2010}. 
We observe that 2D ray graphs are a strict subclass of bipartite MPT graphs
and that they play an interesting role in the structure of neighborhoods of vertices of 
MPT graphs.

\begin{proposition}\label{prop:2D-ray}
2D ray graphs are a strict subclass of bipartite MPT graphs. 
\end{proposition}
\begin{proof}
Notice that, without loss of generality, we may assume that any 2D ray representation of a graph $G$
only uses $\downarrow$ and $\leftarrow$ as the two directions its rays follow.
With this in mind it is easy to see that this representation is in fact a linear
L-system. In particular, we can imagine a line with negative slope that
intersects all the rays and occurs ``below'' and to the ``left'' of any point of
intersection between two rays. Thus, by stopping all rays on this line we have a
linear L-system of $G$. 
Additionally, this inclusion is strict since a 6-cycle is not a 
2D ray graph \cite{Kostochka1998}, but it is easily constructed as a linear L-graph. 
\end{proof}

Recall that MPT graphs have been shown to have $\calO(n^2)$ maximal cliques \cite{CLH2013-point-tol4}. 
Moreover, every complete bipartite graph is a 2D ray graph. 
Thus, $\calO(n^2)$ is a tight bound (up to a multiplicative constant) on the number of maximal cliques in MPT graphs. 

We now consider the neighborhood of a single vertex and observe the following connection to 2D ray graphs and interval graphs. 

\begin{proposition}\label{prop:neighborhood}
If $G$ is an MPT graph and $v$ is a vertex of $G$, then the neighborhood of $v$
can be partitioned into $V_L$ and $V_R$ such that: 
\begin{itemize*}
\item $G[V_L]$ and $G[V_R]$ are interval graphs; and 
\item the bipartite graph induced by the edges connecting vertices from $V_L$ to
$V_R$ is a 2D ray graph. 
\end{itemize*}
\end{proposition}

\begin{proof}
Let $\calL$ be a linear L-system of $G$. We set $V_{L}$ as the neighbors of $v$
whose corner points occur prior to $v$'s corner point and define $V_R$ to be the
remaining neighbors of $v$. Notice that the corner points of vertices in $V_R$
will occur after $v$ in $\calL$. The L-shapes of $V_L$ clearly form an anchored
linear L-system and as such correspond to an interval representation. Thus, by
Proposition~\ref{pro:intervals_as_Ls}, $G[V_L]$ and (similarly) $G[V_R]$ are interval
graphs. Moreover, by considering $\calL$, one can easily see that the bipartite
graph induced by the edges connecting vertices from $V_L$ to $V_R$ is a 2D ray
graph. In particular, the horizontal segments of $V_L$ correspond to
$\leftarrow$ rays and the vertical segments of $V_R$ correspond to $\downarrow$
rays. 
\end{proof}

\subsection{Non-MPT graphs}
\label{sec:non-mpt}

In this section we observe some structural properties of MPT graphs that allow us to identify infinite families of non-MPT graphs. These non-MPT graphs will allow us to compare MPT graphs to planar and permutation graphs. 


\begin{proposition}\label{prop:non-adjacent}
If $G$ is an MPT graph with non-adjacent vertices $u$ and $v$, then $G[N(u) \cap
N(v)]$ is an interval graph. 
\end{proposition}
\begin{proof}
Consider the relative position of $u$ and $v$ in the linear L-system of $G$.
Without loss of generality they must occur as in Figure \ref{fig:non-interval+2}. In each
possibility the corner point of any common neighbour of $u$ and $v$ occurs
in the shaded region; i.e., in every linear L-system of $G$, the L-shapes
corresponding to $N(u) \cap N(v)$ form an anchored linear L-system (anchored to
$v$'s L-shape).  Therefore, by Proposition~\ref{pro:intervals_as_Ls}, $G[N(u) \cap
  N(v)]$ is an interval graph.

\begin{figure}[ht]
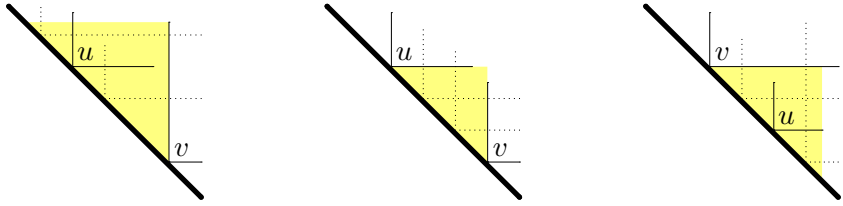
\centering\newxyColor{LightYellow}{1.0 1.0 0.5}{rgb}{}
$\xy/r1pc/:@i@={(0.6,-0.6),(5,-0.6),(5,-5)},0*[LightYellow]\xypolyline{*};
{\ar@{-}@[|2pt] (0,-0.1);(6,-6.1)};{\ar@{.} (1,0);(1,-1)};
{\ar@{.} (1,-1);(6,-1)};{\ar@{-} (2,-0.3);(2,-2)};{\ar@{-} (2,-2);(4.5,-2)};
{\ar@{.} (3,-1.2);(3,-3)};{\ar@{.} (3,-3);(6,-3)};{\ar@{-} (5,-0.6);(5,-5)};
{\ar@{-} (5,-5);(6,-5)};(2.4,-1.6)*{u};(5.4,-4.6)*{v};\endxy$
\qquad\qquad
$\xy/r1pc/:@i@={(2,-2),(5,-2),(5,-5)},0*[LightYellow]\xypolyline{*};
{\ar@{-}@[|2pt] (0,-0.1);(6,-6.1)};{\ar@{-} (2,-0.3);(2,-2)};
{\ar@{-} (2,-2);(4.5,-2)};{\ar@{.} (3,-0.7);(3,-3)};{\ar@{.} (3,-3);(6,-3)};
{\ar@{.} (4,-1.4);(4,-4)};{\ar@{.} (4,-4);(6,-4)};{\ar@{-} (5,-2.5);(5,-5)};
{\ar@{-} (5,-5);(6,-5)};(2.4,-1.6)*{u};(5.4,-4.6)*{v};\endxy$
\qquad\qquad
$\xy/r1pc/:@i@={(2,-2),(5.5,-2),(5.5,-5.5)},0*[LightYellow]\xypolyline{*};
{\ar@{-}@[|2pt] (0,-0.1);(6,-6.1)};{\ar@{-} (2,-0.3);(2,-2)};
{\ar@{-} (2,-2);(6,-2)};{\ar@{.} (3,-1.0);(3,-3)};{\ar@{.} (3,-3);(6,-3)};
{\ar@{-} (4,-2.5);(4,-4)};{\ar@{-} (4,-4);(5.5,-4)};{\ar@{.} (5,-0.5);(5,-5)};
{\ar@{.} (5,-5);(6,-5)};(2.4,-1.6)*{v};(4.4,-3.6)*{u};\endxy$
\caption{The three ways to represent two non-adjacent vertices $u,v$ in a linear
$L$-model. Notice that any common neighbor of $u$ and $v$ must have its corner
point in the shaded region.}
\label{fig:non-interval+2}
\end{figure}
\end{proof}

Notice that Proposition \ref{prop:non-adjacent} is tight. In particular, if one adds a independent set $I$ to an interval graph $G$ such that every element of $I$ is adjacent to every vertex in $G$, then the resulting graph $G'$ is an MPT graph. 
Specifically, by Proposition~\ref{pro:intervals_as_Ls}, $G$ has an anchored linear L-system $\calL$. 
We form a linear L-system for $G'$ as follows. Starting from $\calL$, one simply adds a set of $|I|$ horizontal
segments such that the first one occurs ``just below'' the anchor point of $\calL$, and each subsequent segment occurs ``just below'' the previous segment. Each such segment will intersect every L-shape of in $\calL$ (since they are anchored) and they are disjoint from each other. Thus, this is a linear L-system of $G'$; i.e., $G'$ is an MPT graph. This leads to the following observation regarding minimal forbidden induced subgraphs for MPT graphs. The set of minimal forbidden induced subgraphs of interval graphs is known \cite{LB1962} and is infinite.

\begin{observation}\label{obs:interval-fisc}
If $H$ is a minimal forbidden induced subgraph for interval graphs and $G$ is obtained from $H$ by adding two non-adjacent universal vertices $x$ and $y$ to $H$; i.e., $V(G) = V(H) \cup \{x,y\}$ and $E(G) = E(H) \cup \{zu : z \in \{x,y\}$ and $u \in V(H)\}$, then either $G$,  $G \setminus \{x\}$, or $H$ is a minimal forbidden induced subgraph of MPT graphs.  
\end{observation}

By Proposition \ref{prop:non-adjacent} we see that $K_{2,2,2}$, the graph formed by
taking a 4-cycle together with two non-adjacent
vertices adjacent to each vertex of the cycle, is not an MPT graph. However, it
is easy to see that this graph is a permutation graph as well as a planar graph
(see Figure \ref{fig:k_222}).  Moreover, non planar graphs (e.g., the 5-clique) and non 
permutation graphs (e.g., the graph in Figure
\ref{fig:L-rect-tri}) are both MPT graphs. Thus we have the following observation.

\begin{figure}[ht]
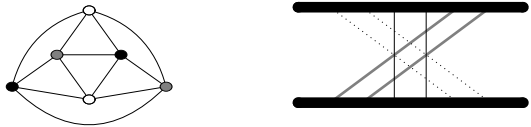
\centering
$\xy/r2pc/:(0,0)*[o]{}="x1";(1,0)*[o]{}="x2";
(0.5,0.7)*[o]{}="x3";(0.5,-0.7)*[o]{}="x4";
(-0.7,-0.5)*[o]{}="x5";(1.7,-0.5)*[o]{}="x6";
{\ar@{-} "x1";"x2"};{\ar@{-} "x1";"x3"};{\ar@{-} "x1";"x4"};{\ar@{-} "x2";"x3"};
{\ar@{-} "x2";"x4"};{\ar@{-} "x1";"x5"};{\ar@{-} "x2";"x6"};
{\ar@{-}@/_/ "x3";"x5"};{\ar@{-}@/^/ "x3";"x6"};{\ar@{-} "x4";"x5"};
{\ar@{-} "x4";"x6"};{\ar@{-}@/_1.2pc/ "x5";"x6"};
"x1"*[o][F**:gray]{\phantom{_s}};"x2"*[o][F**:black]{\phantom{_s}};
"x3"*[o][F-][F**:white]{\phantom{_s}};"x4"*[o][F-][F**:white]{\phantom{_s}};
"x5"*[o][F**:black]{\phantom{_s}};"x6"*[o][F**:gray]{\phantom{_s}};\endxy$
\qquad\qquad
$\xy/r1pc/:{\ar@{.} (1,1.5);(5,-1.5)};{\ar@{-}@[gray]@[|1pt] (1,-1.5);(5,1.5)};
{\ar@{.} (2,1.5);(6,-1.5)};{\ar@{-}@[gray]@[|1pt] (2,-1.5);(6,1.5)};
{\ar@{-} (3,-1.5);(3,1.5)};{\ar@{-} (4,-1.5);(4,1.5)};
{\ar@{-}@[|4pt] (0,-1.5);(7,-1.5)};{\ar@{-}@[|4pt] (0,1.5);(7,1.5)};\endxy$
\caption{A planar drawing of the non-MPT graph $K_{2,2,2}$ together with a
permutation representation of it.}
\label{fig:k_222}
\end{figure}

\begin{observation}\label{obs:incomparable}
The MPT graph class is incomparable with both planar graphs and permutation
graphs.
\end{observation}

The minimal forbidden induced subgraphs of MPT graphs include many more graphs than those built from the graphs non-interval graphs. For example, we will show that the full subdivision of any non-outerplanar graph is also not an MPT graph. The \emph{full subdivision} $G$ of a graph $H$ is the graph obtained from $H$ by subdividing every edge of $G$. It is known that any string representation $\calS$ of the full-subdivision $H$ of a planar graph $G$ is combinatorially equivalent to some planar embedding of $G$ \cite{Ehrlich1976}. In particular, in $\calS$, each edge $e$ of $G$ corresponds to a string $S_e$ which connects exactly the two strings corresponding to vertices incident with $e$ and $S_e$ does not intersect any other strings. From this it is easy to see that the full-subdivision of any non-planar graph is not a string graph \cite{Ehrlich1976}. 

A graph $G$ is an \emph{outer-string} graph when it has a string representation such that, for a fixed circle $C$, every string is contained within $C$ and exactly one endpoint of each string belongs to $C$. It is easy to see that outer-string graphs are a superclass of MPT graphs; i.e., we can easily replicate a linear L-system with an outer-string representation. Moreover, in Lemma \ref{lem:non-outerplanar}, we observe that the full-subdivision of a non-outerplanar graph is not an outer-string graph and, consequently, not an MPT graph. 

\begin{lemma}\label{lem:non-outerplanar}
If $H$ is a non-outerplanar graph and $G$ is the graph obtained from $H$ by subdividing every edge of $H$, then $G$ is not an outer-string graph (i.e., $G$ is not an MPT graph). 
\end{lemma}
\begin{proof}
Since $G$ is not outerplanar, any string representation $\calS$ of $H$ necessarily contains a string $S_v$ such that $S_v$ is contained in a region enclosed by the strings of a set $X$ of non-neighbors of $v$. In particular, it is not possible to draw a circle onto $\calS$ so that both $S_v$ and every string of the vertices in $X$ satisfy the outer-string property. 
\end{proof}

Even with the set of forbidden graphs we have observed, there are yet many more which are not captured (e.g., the complement of a 7-cycle). Thus, a complete description of the minimal forbidden induced subgraphs for MPT graphs remains an open problem. 

\section{Concluding Remarks}

In this paper we have introduced max point-tolerance graphs. 
We have characterized this class and demonstrated inclusions with respect to well-known graph classes. Our results are summarized in Theorems \ref{thm:rel} and \ref{thm:char} below.
We also solved the WIS problem in polynomial time, 2-approximated the clique cover problem in polynomial time, showed the NP-completeness of the coloring problem, and $\log(n)$-approximated the coloring problem in polynomial time. 

Interesting open problems remain for this graph class. 
Perhaps the most interesting is that of recognition. Our characterizations of this graph class provide a variety of ways to approach this problem. 
Several combinatorial optimization problems remain open for this graph class. Two particularly interesting ones are: $k$-coloring (for fixed $k$) and unweighted clique cover. One may also try to improve the $\chi$ binding function as we rely on an existing result regarding coloring axis-aligned rectangles where no pair are in a containment relation.
Additionally, one may be interested to further study the relationships with existing graph classes. 

Another direction of research would be to study \emph{min} point-tolerance
graphs. In particular, just as there are min tolerance graphs and max tolerance
graphs one can consider \emph{min point-tolerance (mPT)} graphs.

\begin{definition}
A graph $G = (V,E)$ is a \emph{min point-tolerance (mPT)} graph if each vertex
$v$ of $G$ can be mapped to a \emph{pointed-interval} $(I_v, p_v)$ where $I_v$
is an interval of $\reals$ and $p_v \in I_v$ such that $uv$ is an edge of $G$
iff either $p_u \in I_v$ or $p_v \in I_u$.
\end{definition}

It should be noted that this graph class is referred to as \emph{point-core bi-tolerance} graphs and further information can be found in Chapter 5 of \cite{GT-Tolerance2004}. 
Additionally, there is a directed graph class utilizing
this definition, namely, the interval catch
digraphs~\cite{Prisner1989-CatchDigraphs,Prisner1994-CatchDigraphs} mentioned in
the introduction. The min point-tolerance graphs are precisely the undirected
graphs underlying interval catch digraphs. In contrast, max point-tolerance graphs
are precisely the undirected graphs underlying the bi-directed edges of interval
catch digraphs.

\begin{theorem}\label{thm:rel}
The max point-tolerance graph class strictly includes interval graphs, outerplanar graphs, and 2D ray graphs. 


\end{theorem}

\begin{theorem}\label{thm:char}
For a graph $G = (V,E)$, the following are equivalent:

\begin{itemize*}
\item $G$ is a max point-tolerance graph. 

\item $G$ is a linear L-graph (equivalently, a linear rectangle-graph or a
linear right-triangle-graph). 

\item The vertices of $G$ can be ordered by $<$ so that for every $u,v,w,x \in
V(G)$, if $u<v<w<x$ and $uw,vx \in E(G)$, then $vw$ is an edge of $G$. 

\item There are two interval graphs $H_1 = (V,E_1)$ and $H_2 = (V,E_2)$ such
that $E = E_1 \cap E_2$ and the vertices of $G$ can be ordered by $<$ so that
for every $u<v<w$ if $uw \in E_1$ then $uv \in E_1$ and if $uw \in E_2$ then $wv
\in E_2$. 

\item $G$ is a cyclic segment graph.

\item There is an interval catch digraph $D = (V,A)$ such that the bi-directed
arcs of $D$ are precisely the edges of $G$. 
\end{itemize*}

\end{theorem}

\section*{Acknowledgments}

Steven Chaplick was supported by NSERC and by the ESF GraDR 
EUROGIGA grant as project GACR GIG/11/E023.  
Stefan Felsner was partially supported by the ESF EuroGIGA projects Compose and GraDR. 
Magn\'{u}s M. Halld\'{o}rsson was supported by the Icelandic Research Fund grant-of-excellence 120032011.
Thomas Hixon was supported by BMS (Berlin Mathematical School).
Juraj Stacho gratefully acknowledges support from EPSRC, grant EP/I01795X/1, and from the Centre for 
Discrete Mathematics and its Applications (DIMAP), which is partially funded 
by EPSRC award EP/D063191/1.

\section*{References}

\end{document}